\documentclass[a4paper,onecolumn,accepted=2023-08-21]{quantumarticle}
\pdfoutput=1
\usepackage{graphicx}
\usepackage[utf8]{inputenc}
\usepackage[english]{babel}
\usepackage[T1]{fontenc}
\usepackage[numbers]{natbib}
\usepackage{amsthm}
\usepackage{amssymb}
\usepackage{braket}
\usepackage{bm}
\usepackage{here}
\usepackage{hyperref}
\usepackage{ascmac}
\usepackage{comment}
\usepackage{booktabs}
\usepackage{mathtools}
\usepackage{framed}
\usepackage{color}
\newtheorem{theorem}{Theorem}
\newtheorem{lemma}{Lemma}
\newtheorem{prop}{Proposition}
\newtheorem{cor}{Corollary}
\newtheorem{definition}{Definition}
\newtheorem*{remark}{Remark}

\newcommand{\relmiddle}[1]{\mathrel{}\middle#1\mathrel{}}
\def\({\left(}

\def\){\right)}

\begin{document}
\title{Refined finite-size analysis of binary-modulation continuous-variable quantum key distribution}
\author{Takaya Matsuura}
 \email{takaya.matsuura@rmit.edu.au}
 \affiliation{Department of Applied Physics, Graduate School of Engineering, The University of Tokyo, 7-3-1 Hongo, Bunkyo-ku, Tokyo 113-8656, Japan} 
 \affiliation{Centre for Quantum Computation \& Communication Technology, School of Science, RMIT University, Melbourne VIC 3000, Australia}
 \orcid{0000-0003-4164-4307}
 \author{Shinichiro Yamano}
 \affiliation{Department of Applied Physics, Graduate School of Engineering, The University of Tokyo, 7-3-1 Hongo, Bunkyo-ku, Tokyo 113-8656, Japan} 
 \orcid{0000-0003-3245-5344}
 \author{Yui Kuramochi}
 \affiliation{Department of Physics, Faculty of Science, Kyushu University, 744 Motooka, Nishi-ku, Fukuoka, Japan}  
 \orcid{0000-0003-0512-5446}
\author{Toshihiko Sasaki}
 \affiliation{Department of Applied Physics, Graduate School of Engineering, The University of Tokyo, 7-3-1 Hongo, Bunkyo-ku, Tokyo 113-8656, Japan} 
 \affiliation{Photon Science Center, Graduate School of Engineering, The University of Tokyo, 7-3-1 Hongo, Bunkyo-ku, Tokyo 113-8656, Japan}
 \orcid{0000-0003-0745-6791}
\author{Masato Koashi}
 \affiliation{Department of Applied Physics, Graduate School of Engineering, The University of Tokyo, 7-3-1 Hongo, Bunkyo-ku, Tokyo 113-8656, Japan} 
 \affiliation{Photon Science Center, Graduate School of Engineering, The University of Tokyo, 7-3-1 Hongo, Bunkyo-ku, Tokyo 113-8656, Japan}
 \orcid{0000-0002-1952-6470}

\maketitle

\begin{abstract}
    Recent studies showed the finite-size security of binary-modulation CV-QKD protocols against general attacks. 
    However, they gave poor key-rate scaling against transmission distance.  Here, we extend the security proof based on complementarity, which is used in the discrete-variable QKD, to the previously developed binary-modulation CV-QKD protocols with the reverse reconciliation under the finite-size regime and obtain large improvements in the key rates.  Notably, the key rate in the asymptotic limit scales linearly against the attenuation rate, which is known to be optimal scaling but is not achieved in previous finite-size analyses.  This refined security approach may offer full-fledged security proofs for other discrete-modulation CV-QKD protocols.
\end{abstract}

\section{Introduction}
Quantum key distribution (QKD) \cite{BB84} enables two remote parties to share identical random secret bits that are secure against arbitrary eavesdropping allowed under the law of quantum mechanics.  
Using the one-time pad \cite{Shannon1949} with random secret bits shared by the QKD, we can realize the information-theoretic security for bipartite communication. 
Nowadays, there is increasing interest in implementing QKD in the real world.  
Among other things, continuous-variable (CV) QKD \cite{Ralph1999, Hillery2000, Gottesman2001, Cerf2001, GG02, Grosshans2003, hetero04} has the advantages of low-cost implementation and high bit rates at short distances.  Furthermore, it is relatively easy to multiplex several QKD channels via wavelength division multiplexing and co-propagate them with the classical telecommunication because homodyne and heterodyne detectors used in CV-QKD protocols do not require a low-temperature environment and have good wavelength selectivity.  
The (single-)photon detectors used in discrete-variable (DV) QKD, on the other hand, typically require a low-temperature environment for stable operation and a high-quality frequency filter for the wavelength division multiplexing.

The main problem with CV-QKD protocols is the difficulty in establishing a complete security proof.
Compared to the DV-QKD protocols, most of which have complete security proofs even in the finite-size regime, almost all the CV-QKD protocols only have asymptotic security proofs \cite{Silberhorn2002, two_state_Lutken, Garcia2009, Leverrier2009, Leverrier2011, Kaur2021, four_state_Leverrier, four_state_Lutken, liu2021, Upadhyaya2021, denys2021} or security proofs against collective attacks \cite{leverrier2010, Leverrier2015, samsonov2020, Papanastasiou2021continuous, Papanastasiou2021security,mountogiannakis2021,Kanitschar2023}.  
There are, however, some results for the composable finite-size security against general attacks.  
One is for the protocol using the two-mode squeezed vacuum state \cite{Furrer2012, Furrer2014, Furrer2014reverse}, whose security proof is based on the entropic uncertainty relation in the infinite dimensional systems \cite{Berta2016}.  
Unfortunately, this protocol is difficult to implement and has a key-rate that scales poorly as transmission distance grows.   Another is for the $U(N)$-symmetric protocol that uses coherent states with their complex amplitudes modulated according to a Gaussian distribution \cite{Leverrier2013, Leverrier2015, Gaussian_unitary}.  
The security proof for this type of protocol utilizes the de Finetti reduction theorem \cite{leverrier2009security,leverrier2017} to the i.i.d.\ case.  
This methodology has proved the security of several $U(N)$-invariant CV-QKD protocols \cite{Lupo2018, Ghorai2019}.  
However, in practice, ideal Gaussian modulation cannot be implemented and should be approximated by a finite number of coherent states.  
It turns out that an overwhelming number of coherent states is needed to directly approximate the Gaussian ensemble for the security condition to be satisfied \cite{Jouguet2012, Lupo2020}.  
If we try to mitigate the required number, additional assumptions are needed, which makes it difficult to apply in the finite-size regime \cite{Kaur2021}.

The alternative approach \cite{matsuura2021, Yamano2022} is to consider discrete-modulation CV QKD (See also Ref.~\cite{Bauml2023} for recent alternative approach). 
References~\cite{matsuura2021, Yamano2022} show the finite-size security against general attacks for a binary-modulation protocol.  
It can also take into account the discretization of the signal processing, such as binned homodyne and heterodyne measurements (see also Ref.~\cite{lupo2021} for this topic).  
Although it offers such a flexible framework for security proof against general attacks, the obtained key rate has very poor scaling against transmission distance.  
A possible reason for this bad performance is the fact that its security proof is based on entanglement distillation \cite{Shor_Preskill, Lo1999}. 
It is known that the security proof based on entanglement distillation is too stringent in general for secure key distribution.
There are alternative types of security proofs \cite{renner2008, complementarity, tsurumaru2020equivalence} that can be applied to general cases.  
In particular, for CV-QKD protocols, the security proof based on the reverse reconciliation often provides better performance than that based on the direct reconciliation \cite{Silberhorn2002}, which may be unattainable by a security proof based on entanglement distillation due to its symmetric nature between the sender and the receiver in the security proof.

\bigskip
\noindent {\bf Contributions of this paper.}
In this article, we aim to develop another approach to carry out the finite-size security proof for the discrete-modulation CV QKD against general attacks.  The approach should be able to exploit the benefit of the reverse reconciliation.
To do it concretely, we develop refined security proofs based on the reverse reconciliation for the binary-modulation CV-QKD protocols proposed in Refs.~\cite{matsuura2021, Yamano2022}, i.e., the protocol in which the sender Alice performs BPSK-type modulation according to her randomly generated bit and the receiver Bob performs homodyne measurement, heterodyne measurement, and trash randomly \cite{matsuura2021}, or performs heterodyne measurement followed by a random selection of the post-processing of the outcome \cite{Yamano2022}.  We use the same apparatuses and setups as those in Refs.~\cite{matsuura2021, Yamano2022} but slightly change the protocols.  To refine the security proofs, we use an approach based on complementarity \cite{complementarity, Hayashi_Tsurumaru} under the reverse reconciliation, which is more general than the one based on entanglement distillation \cite{tsurumaru2020leftover} and treats Alice and Bob asymmetrically in the security proof.  
In these refined security proofs, we have degrees of freedom that did not appear in the previous analyses.  By setting these degrees of freedom to be optimal in the pure-loss channel \cite{pirandola2017}, we obtain a significant improvement in the key gain rates; in fact, the asymptotic key rates of the protocols scale linearly with regards to the attenuation rate of the pure-loss channel, which is known to be the optimal scaling for the one-way QKD \cite{pirandola2017}.  This shows that we can exploit the benefit of using the reverse reconciliation in the approach based on complementarity.  Although the protocols are still fragile against the excess noise, this approach itself may be a step towards the full-fledged security proofs for discrete-modulation CV QKD.   

\bigskip
\noindent {\bf Organization of this paper.}
The article is organized as follows.  
In Section \ref{sec:security_proof}, we provide the refined security proofs based on complementarity \cite{complementarity} for protocols that use the same experimental setups as proposed in Refs.~\cite{matsuura2021, Yamano2022}.  
The section is further divided into four parts.  
The first part~\ref{sec:actual_protocol} defines the actual protocols, which are almost the same as the ones in Refs.~\cite{matsuura2021, Yamano2022}.  The second part~\ref{sec:virtual_estimation} develops virtual protocols for the complementarity approach \cite{complementarity} and evaluates the necessary amount of privacy amplification through estimation protocols.  In the third part \ref{sec:phase_error_operator}, we derive an explicit form of the phase error operator defined by the virtual procedure of the previous part.  
In the fourth part \ref{sec:finite_size_analysis}, we finish the finite-size security proof by developing operator inequalities.  
In Section \ref{sec:numerical_simulations}, we numerically demonstrate the improved performance of the protocols with our refined security proof.  
Finally, in Section \ref{sec:discussion}, we wrap up our article by discussing future works and open problems. 

\section{Security proof} \label{sec:security_proof}

In this section, we define two binary-modulation CV-QKD protocols that are closely related to the ones proposed in Refs.~\cite{matsuura2021, Yamano2022}, and present their security proofs based on the reverse reconciliation.  
The definition of the (composable) security is the same as that in Ref.~\cite{matsuura2021} and given in Appendix~\ref{sec:definition_security}. 

\subsection{Definition of the protocol} \label{sec:actual_protocol}

The setups of the protocols are illustrated in Fig.~\ref{fig:setups}.
In the following, a random number is denoted with a hat such as $\hat{\cdot}$.  For the places where the slash ``/'' is used, one can adopt either its left-hand side or right-hand side depending on which of ``Homodyne protocol'' or ``Heterodyne protocol'' defined in Fig.~\ref{fig:setups} one chooses.  Note that Homodyne protocol is the same as the protocol proposed in Ref.~\cite{matsuura2021} except for the definition of $f_{\rm suc}(x)$ as well as the way of bit error correction, and Heterodyne protocol is the same as the protocol proposed in Ref.~\cite{Yamano2022} except for the additional trash round as well as the way of bit error correction. 

\begin{figure}
    \centering
    \includegraphics[width=0.97\linewidth]{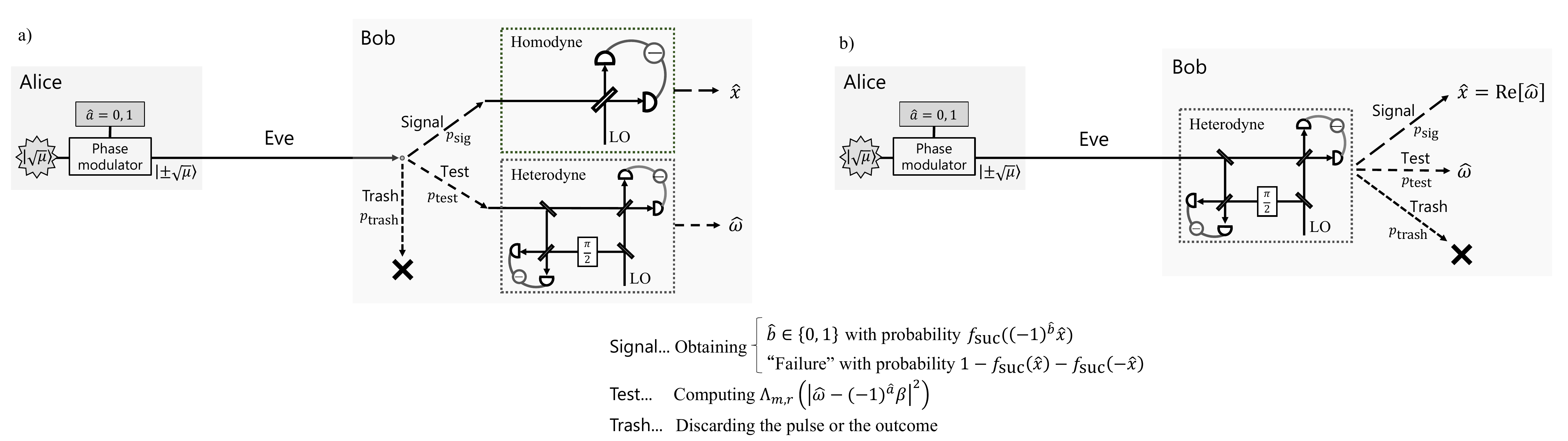}
    \caption{Setups of the protocols.  In both protocols, the sender Alice modulates the optical phase of a laser pulse prepared in a coherent state $\ket{\mu}$ with $0$ or $\pi$ according to her random bit $\hat{a}=0$ or $1$.  (a) ``Homodyne protocol'', which is similar to the protocol proposed in Ref.~\cite{matsuura2021}.  In this protocol, the receiver Bob randomly switches three types of measurements according to probabilities $p_{\rm sig}$, $p_{\rm test}$, and $p_{\rm trash}$, respectively.  In ``Signal'', Bob performs homodyne measurement and obtains the outcome $\hat{x} \in \mathbb{R}$.  Then, he obtains $\hat{b}\in\{0,1\}$ with probability $f_{\rm suc}\bigl((-1)^{\hat{b}}\hat{x}\bigr)$, respectively, or announces ``Failure'' with probability $1-f_{\rm suc}(\hat{x})-f_{\rm suc}(-\hat{x})$.  In ``Test'', Bob performs heterodyne measurement and obtains the outcome $\hat{\omega} \in \mathbb{C}$.  Then, he computes $\Lambda_{m,r}\bigl(|\hat{\omega}-(-1)^{\hat{a}}|^2\bigr)$ with Alice's bit $\hat{a}$ announced.  In ``Trash'', Bob discards the received optical pulse and produces no outcome.  (b) ``Heterodyne protocol'', which is similar to the protocol proposed in Ref.~\cite{Yamano2022}.  In this protocol, the receiver Bob performs heterodyne measurement, obtains the outcome $\hat{\omega}$, and randomly switches three types of post-processings according to probabilities $p_{\rm sig}$, $p_{\rm test}$, and $p_{\rm trash}$, respectively.  In ``Signal'', Bob defines $\hat{x}=\Re[\hat{\omega}]$ and follows the same procedure of obtaining the bit $\hat{b}$ or ``Failure'' as in Homodyne protocol.  In ``Test'', Bob follows the same procedure of computing $\Lambda_{m,r}\bigl(|\hat{\omega}-(-1)^{\hat{a}}|^2\bigr)$ as in Homodyne protocol.  In ``Trash'', Bob discards the outcome $\hat{\omega}$.}
    \label{fig:setups}
\end{figure}

Prior to the protocol, Alice and Bob determine the number $N$ of total rounds, the acceptance probability function $f_{\rm suc}(x)\ (x \in \mathbb{R})$ of the homodyne/heterodyne measurement satisfying $f_{\rm suc}(x)+f_{\rm suc}(-x)\leq 1$, an odd integer $m$ and a real $r$ for the test function $\Lambda_{m,r}(\nu)\coloneqq e^{-r \nu}(1+r)L_m^{(1)}((1+r)\nu)$ with $L_m^{(1)}$ being the associated Laguerre polynomial \cite{matsuura2021}, and the protocol parameters $(\mu, p_{\rm sig}, p_{\rm test}, p_{\rm trash}, \beta, s, \kappa, \gamma)$ satisfying $p_{\rm sig}+p_{\rm test}+p_{\rm trash}=1$ and $\beta < \sqrt{\mu}$, where all the parameters are positive.  Alice and Bob then run the protocol described in Box 1, where the outcome $\hat{x}\in\mathbb{R}$ of a homodyne measurement and the outcome $\hat{\omega}\in\mathbb{C}$ of a heterodyne measurement shall be normalized such that their variances are $\braket{(\varDelta \hat{x})^2}=1/4$ and $\braket{\bigl(\varDelta \Re[\hat{\omega}]\bigr)^2}=\braket{\bigl(\varDelta \Im[\hat{\omega}]\bigr)^2}=1/2$ for the vacuum.
Unless aborted or $\hat{N}^{\rm suc}=0$, the protocol generates a shared final key of length 
\begin{equation} 
  \hat{N}^{\rm fin} = \max\left\{\hat{N}^{\rm suc} - \Bigl\lceil \hat{N}^{\rm suc} h\!\(\min\bigl\{U(\hat{F}, \hat{N}^{\rm trash})/\hat{N}^{\rm suc}, 0.5\bigr\}\) \Bigr\rceil - s, 0\right\},
  \label{eq:final_key_length}
\end{equation}
where $\lceil\cdot\rceil$ is the ceiling function, $h(x)$ is the binary entropy function defined as
\begin{equation}
    h(x)\coloneqq 
        -x \log_2(x)-(1-x) \log_2(1-x),
\end{equation}
and the function $U(\hat{F}, \hat{N}^{\rm trash})$ will be specified later.

\smallskip

\begin{framed}
{\noindent \bf Box 1: Actual protocol}

\vspace{0.15cm}
\begin{enumerate}
  \setlength{\leftskip}{0.5cm}
  \item For each of the $N$ rounds, Alice generates a random bit $\hat{a}\in\{0,1\}$ and sends an optical pulse $\tilde{C}$ in a coherent state with amplitude $(-1)^{\hat{a}} \sqrt{\mu}$ to Bob. Bob receives an optical pulse $C$ for each round.
  \item For the received pulse $C$ in each round, Bob chooses a label from $\rm \{signal, test, trash\}$ with probabilities $p_{\rm sig},p_{\rm test}$, and $p_{\rm trash}$, respectively, and announces it. According to the label, Alice and Bob do one of the following procedures.
  \begin{itemize}
  \item[\textrm{[}signal\textrm{]}] Bob performs a homodyne/heterodyne measurement on the received optical pulse $C$ and obtains an outcome $\hat{x} \in \mathbb{R}$.  (For the heterodyne measurement, $\hat{x}$ is defined as the real part of the outcome $\hat{\omega} \in \mathbb{C}$.)  Bob defines a sifted-key bit $\hat{b}$ as $\hat{b}=0$ with a probability $f_{\rm suc}(\hat{x})$ and $\hat{b}=1$ with a probability $f_{\rm suc}(-\hat{x})$.  When Bob has defined his sifted key bit, he announces ``success'', and otherwise, he announces ``failure''.  In the case of a success, Alice (resp.~Bob) records a bit $\hat{a}$ $(\hat{b})$.
  \item[\textrm{[}test\textrm{]}] Bob performs a heterodyne measurement on the received optical pulse $C$ and obtains an outcome $\hat{\omega}$.  Alice announces her bit $a$.  Bob calculates the value of $\Lambda_{m,r}(|\hat{\omega} -(-1)^{\hat{a}}\beta|^2)$. 
  \item[\textrm{[}trash\textrm{]}] Alice and Bob produce no outcomes.
  \end{itemize}
  \item We refer to the numbers of ``success'' and ``failure'' signal rounds, test rounds, and trash rounds as $\hat{N}^{\rm suc}, \hat{N}^{\rm fail}, \hat{N}^{\rm test}$, and $\hat{N}^{\rm trash}$, respectively.  ($N= \hat{N}^{\rm suc}+\hat{N}^{\rm fail}+\hat{N}^{\rm test}+\hat{N}^{\rm trash}$ holds by definition.) 
  Bob calculates the sum of $\Lambda_{m,r}(|\hat{\omega} -(-1)^{\hat{a}}\beta|^2)$ obtained in the $\hat{N}^{\rm test}$ test rounds, which is denoted by $\hat{F}$.
  \item For error correction, they use $H_{\rm EC}$-bits of encrypted communication consuming a pre-shared secret key to do the following. 
  According to (the upper bound on) the bit error rate $e_{\rm qber}$, Bob randomly chooses an error-correcting code and sends it with the $H_{\rm EC}$-bits syndrome to Alice.  Alice reconciles her sifted key accordingly. 
  \item Bob computes and announces the final key length $\hat{N}^{\rm fin}$ according to Eq.~\eqref{eq:final_key_length}.
  Alice and Bob apply privacy amplification to obtain the final key. 
\end{enumerate}
\end{framed}

For simplicity, we omitted the bit-error-sampling rounds in the above protocol.  
To satisfy the required correctness $\varepsilon_{\rm cor}$ for the final key (see Appendix~\ref{sec:definition_security}), Alice and Bob randomly insert $N_{\rm smp}$ sampling rounds among $N$ rounds in which Bob performs the same measurement as that of the signal round and estimate an upper bound $e_{\rm qber}$ on the bit error rate.
Let $\hat{N}_{\rm smp}^{\rm suc}$ be the number of ``success'' in $N_{\rm smp}$ sampling rounds, and let $\hat{E}_{\rm obs}$ be the number of discrepancies between Alice's and Bob's bits observed in the ``success'' sampling rounds.
Then, Bob sets $e_{\rm qber}$ to 
\begin{equation}
    e_{\rm qber} = \left.\left(\tilde{M}_{\hat{N}^{\rm suc}+\hat{N}_{\rm smp}^{\rm suc},\hat{N}_{\rm smp}^{\rm suc},\varepsilon_{\rm cor}/2}(\hat{E}_{\rm obs}) - \hat{E}_{\rm obs}\right) \right/ \hat{N}^{\rm suc},
    \label{eq:upper_bit_error_rate}
\end{equation}
where the function $\tilde{M}_{N,n,\epsilon}$ is defined in Eq.~\eqref{eq:upper_bound_sampling} in Appendix~\ref{sec:estimate_num_of_err}.
The proof that this definition of $e_{\rm qber}$ upper-bounds the actual bit error rate with probability no smaller than $1-\varepsilon_{\rm cor}/2$ is also shown in Appendix~\ref{sec:estimate_num_of_err}. 
The required amount $H_{\rm EC}$ of the error syndrome Bob sends to Alice in the bit error correction depends on the error correction method; here we assume 
\begin{equation}
    H_{\rm EC} = \hat{N}^{\rm suc} \left[f\, h(\min\{e_{\rm qber},0.5\}) + (1 - f)\right], 
    \label{eq:bit_error_fraction}
\end{equation} 
where $f\in[0,1]$ denotes an error correction efficiency \cite{lodewyck2007,jouguet2014,four_state_Leverrier,four_state_Lutken,liu2021,wen2021} for the error correction to succeed with the probability no smaller than $1-\varepsilon_{\rm cor}/2$.
The net key gain $\hat{G}$ per pulse is thus given by
\begin{equation}
    \hat{G} = (\hat{N}^{\rm fin} - H_{\rm EC}) / (N + N_{\rm smp}).
\end{equation}
Here, we do not use verification in the post-processing, unlike Refs.~\cite{matsuura2021, Yamano2022}, due to the subtleties to incorporate it in our security proof.
The acceptance probability $f_{\rm suc}(x)$ should be chosen to post-select the rounds with larger values of $x$, for which the bit error probability is expected to be lower.  
The definition of $f_{\rm suc}(x)$ in this article follows Ref.~\cite{Yamano2022} and is slightly more general than that of Ref.~\cite{matsuura2021}.  (Note that Ref.~\cite{matsuura2021} can also use this definition of $f_{\rm suc}(x)$.)  It is ideally a step function with a threshold $x_{\mathrm{th}}(>0)$, but our security proof applies to any form of $f_{\rm suc}(x)$.  The test function $\Lambda_{m,r}(\nu)$ is the same as the one defined in Ref.~\cite{matsuura2021} where it is shown to satisfy
\begin{equation}
    \mathbb{E}_{\rho}[\Lambda_{m,r}(|\hat{\omega}-\beta|^2)]\leq \bra{\beta}\rho\ket{\beta}
    \label{eq:lower_fid_coherent}
\end{equation}
for any odd integer $m$, positive real $r$, and density operator $\rho$ (see Appendix~\ref{sec:fidelity_lower_bound}).
The parameter $\beta$ is typically chosen to be $\sqrt{\eta\mu}$ with $\eta$ being a nominal transmissivity of the quantum channel, while the security proof itself holds for any choice of $\beta$.  
The parameter $s$ is related to the overall security parameter in the security proof below.

\subsection{Construction of a virtual protocol and reduction to an estimation protocol} \label{sec:virtual_estimation}

We determine a sufficient amount of the privacy amplification according to the complementarity, or in other words, the phase error correction \cite{complementarity, Hayashi_Tsurumaru}, which has been widely used for the DV-QKD protocols.  
We aim at showing the secrecy of Bob's final key against the adversary Eve (see Appendix~\ref{sec:definition_security}).  
To do so, we consider a virtual protocol in which Bob has a qubit for each success signal round such that the outcome of the $Z$-basis measurement on it is equivalent to his sifted key bit $\hat{b}$.  
Alice can do arbitrary quantum operations in the virtual protocol as long as all the statistics and available information to the adversary Eve are the same as those in the actual protocol.  
Then, after Bob's $Z$-basis measurement on the qubit, the reduced classical-quantum state between Bob and Eve in the virtual protocol is the same as that in the actual protocol.

In the following, we explicitly describe the virtual protocol.
For Alice, we introduce a qubit $A$ and assume that she entangles it with an optical pulse $\tilde{C}$ in a state
\begin{equation}
    \ket{\Psi}_{A\tilde{C}}\coloneqq \frac{\ket{0}_A\ket{\sqrt{\mu}}_{\tilde{C}}+\ket{1}_A\ket{-\sqrt{\mu}}_{\tilde{C}}}{\sqrt{2}},
    \label{eq:prepared_state}
\end{equation}
where $\ket{\omega}_{\tilde{C}}$ with $\omega\in\mathbb{C}$ denotes the coherent state with the amplitude $\omega$, which is defined as
\begin{equation}
    \ket{\omega}_{\tilde{C}} \coloneqq e^{-\frac{|\omega|^2}{2}} \sum_{n=0}^{\infty} \frac{\omega^{n}}{\sqrt{n!}} \ket{n}_{\tilde{C}}.
\end{equation}
Then, the optical pulse $\tilde{C}$ emitted by Alice is in the same state as that in the actual protocol.
For Bob, we construct a process of probabilistically converting the received optical pulse $C$ to a qubit $B$, which can be regarded as a coherent version of Bob's signal measurement.  For Homodyne protocol, consider a map $\mathcal{K}_{C\to B}^{\rm hom}$ defined as \cite{matsuura2021}
\begin{equation}
    \label{eq:I_homo}
    \mathcal{K}_{C\to B}^{\rm hom}(x)(\rho_{C})\coloneqq  K^{\rm hom}_{\rm suc}(x) \,\rho_{C} \, \bigl(K^{\rm hom}_{\rm suc}(x)\bigr)^{\dagger}
\end{equation}
with
\begin{equation}
    K^{\rm hom}_{\rm suc}(x) \coloneqq \sqrt{f_{\rm suc}(x)} \bigl(\ket{0}_{B}\!\bra{x}_C + \ket{1}_{B} \!\bra{-x}_C \bigr),
\end{equation}
where $\bra{x}$ maps a state vector to the value of its wave function at $x$; i.e., for a coherent state vector $\ket{\omega}$, $\bra{x}$ acts as
\begin{equation}
    \braket{x|\omega} = \left(\frac{2}{\pi}\right)^{\frac{1}{4}} \exp\!\left[-(x-\omega_r)^2 + 2i\omega_i x - i\omega_r \omega_i\right],\label{eq:wave_func_coherent}
\end{equation}
where $\omega = \omega_r + i\omega_i$ with $\omega_r, \omega_i\in\mathbb{R}$.
Let $\Pi_{\rm ev(od)}$ denote a projection operator onto the subspace of even(odd) photon numbers.  Since $\bra{x}(\Pi_{\rm ev}-\Pi_{\rm od})=\bra{-x}$ holds, we have
\begin{align}
    \label{eq:K_x}
      K^{\rm hom}_{\rm suc}(x) = \sqrt{2 f_{\rm suc}(x)} \bigl(\ket{+}_{B}\!\bra{x}_C \Pi_{\rm ev} + \ket{-}_{B} \!\bra{x}_C \Pi_{\rm od}\bigr). 
\end{align}
This defines an instrument $\mathcal{I}_{C\to B}^{\rm hom}$ for the process of producing the outcome $\hat{x}$ and leaving $C$ in a post-measurement state; i.e., given a measurable set $\Delta\subseteq\mathbb{R}$, the unnormalized post-measurement state is given by
\begin{equation}
    \mathcal{I}_{C\to B}^{\rm hom}(\Delta)(\rho_{C}) = \int_{\Delta} dx\; \mathcal{K}_{C\to B}^{\rm hom}(x)(\rho_{C})
    \label{eq:instrument_hom}
\end{equation}
with $\mathrm{Tr}[\mathcal{I}_{C\to B}^{\rm hom}(\Delta)(\rho_{C})]$ being a probability of ``success'' signal event with the outcome $\hat{x}\in\Delta$.
Similarly, for Heterodyne protocol, consider a map $\mathcal{K}_{C\to B}^{\rm het}$ defined as \cite{Yamano2022}
\begin{equation}
    \label{eq:I_hetero}
    \mathcal{K}_{C\to B}^{\rm het}(\omega)(\rho_{C})\coloneqq  K^{\rm het}_{\rm suc}(\omega)\,\rho_{C} \bigl(K^{\rm het}_{\rm suc}(\omega)\bigr)^{\dagger} 
\end{equation}
with 
\begin{equation}
    \label{eq:K_omega}
    \begin{split}
        K^{\rm het}_{\rm suc}(\omega) &\coloneqq \sqrt{\frac{f_{\rm suc}(\omega_r)}{\pi}} \bigl(\ket{0}_{B}\!\bra{\omega}_C + \ket{1}_{B} \!\bra{-\omega}_C \bigr)\\
    &= \sqrt{\frac{2f_{\rm suc}(\omega_r)}{\pi}} \left(\ket{+}_{B}\!\bra{\omega}_C \Pi_{\rm ev} + \ket{-}_{B}\!\bra{\omega}_C \Pi_{\rm od}  \right),
    \end{split}
\end{equation}
where $\ket{\omega}$ denotes a coherent state vector and $\omega = \omega_r + i\omega_i$ with $\omega_r, \omega_i\in\mathbb{R}$.
Similarly to Homodyne protocol, we can define an instrument $\mathcal{I}_{C\to B}^{\rm het}$ composed of the heterodyne outcome and the (unnormalized) post-measurement state, which is given by
\begin{equation}
    \mathcal{I}_{C\to B}^{\rm het}(\Delta')(\rho_{C}) = \int_{\Delta'} d\omega_r d\omega_i\; \mathcal{K}_{C\to B}^{\rm het}(\omega)(\rho_{C}),
    \label{eq:instrument_het}
\end{equation}
where $\Delta'\subseteq \mathbb{R}^2$ is a measurable set.
If Bob measures the qubit $B$ in the $Z$ basis after the instrument~\eqref{eq:instrument_hom} (resp.~\eqref{eq:instrument_het}), he obtains the same sifted key bit with the same probability as in the actual protocol when $\hat{x}\in \Delta$ (resp.~$\hat{\omega}\in\Delta'$) \cite{matsuura2021, Yamano2022}.

At this point, one has a degree of freedom to perform quantum operations on the system $AB$ for each outcome $\hat{x}$ (resp.~$\hat{\omega}$) as long as it does not change the $Z$-basis value of the qubit $B$.  This is because we aim at showing the secrecy of Bob's final key against the adversary Eve with Alice's system traced out.
Thus, after applying the map ${\cal K}_{C\to B}^{\rm hom}$ (resp.~${\cal K}_{C\to B}^{\rm het}$), we assume that Alice and Bob perform a controlled isometry $V_{B; A\to R}^{\rm hom}(x)$ (resp.~$V_{B; A\to R}^{\rm het}(\omega)$) of the form
\begin{align}
    V_{B;A\to R}^{\rm hom}(x)&\coloneqq \left[\ket{0}\!\bra{0}_B\otimes V^{(0)}_{A\to R}(x) + \ket{1}\!\bra{1}_B\otimes V^{(1)}_{A\to R}(x)\right] \text{C-}X_{BA}\label{eq:controlled-isometry_1} \\
    V_{B;A\to R}^{\rm het}(\omega)&\coloneqq \left[\ket{0}\!\bra{0}_B\otimes V'^{(0)}_{A\to R}(\omega) + \ket{1}\!\bra{1}_B\otimes V'^{(1)}_{A\to R}(\omega)\right]\text{C-}X_{BA},
    \label{eq:controlled-isometry_2}
\end{align}
where $\text{C-}X_{BA}\coloneqq \ket{0}\!\bra{0}_B \otimes I_A + \ket{1}\!\bra{1}_B \otimes X_A$ denotes the Controlled-NOT gate and $V^{(j)}_{A\to R}(x)$ (resp.~$V'^{(j)}_{A\to R}(\omega)$) for $j=0,1$ denotes an isometry from the system $A$ to another system $R$ that is no smaller than $A$ \footnote{Here, a subtlety for using the verification comes in.  In order to know whether verification succeeds or not, Alice has to confirm the syndrome bits for the verification.  However, this procedure may not commute with the action of $V_{B; A\to R}^{\rm hom}(x)$ (resp.~$V_{B; A\to R}^{\rm het}(\omega)$).  We do not currently have a method to evaluate how much the verification affects the secrecy condition.}.  If $V^{(j)}_{A\to R}(x)$ (resp.~$V'^{(j)}_{A\to R}(\omega)$) is an identity, then the analysis reduces to the previous results \cite{matsuura2021, Yamano2022}.  Let ${\cal V}_{B; A\to R}^{\rm hom}(x)$ (resp.~${\cal V}_{B; A\to R}^{\rm het}(\omega)$) be the CPTP map defined as 
\begin{align}
    {\cal V}_{B; A\to R}^{\rm hom}(x)(\rho_{AB}) &= V_{B; A\to R}^{\rm hom}(x) \rho_{AB} \bigl(V_{B; A\to R}^{\rm hom}(x)\bigr)^{\dagger}, \\
    {\cal V}_{B; A\to R}^{\rm het}(\omega)(\rho_{AB}) &= V_{B; A\to R}^{\rm het}(\omega) \rho_{AB} \bigl(V_{B; A\to R}^{\rm het}(\omega)\bigr)^{\dagger}.
\end{align}
The composition of the map ${\cal V}_{B; A\to R}^{\rm hom}(x)$ and the map~\eqref{eq:I_homo} (resp.~map ${\cal V}_{B; A\to R}^{\rm het}(\omega)$ and the map~\eqref{eq:I_hetero}) with Alice's system traced out at the end defines a quantum operation ${\cal F}^{\rm hom}_{AC\to B}$ (resp.~${\cal F}^{\rm het}_{AC\to B}$) that (probabilistically) outputs Bob's qubits for his sifted key as
\begin{align}
    {\cal F}^{\rm hom}_{AC\to B}(\rho_{AC}) &= \int_{-\infty}^{\infty} dx\; {\cal K}'^{\,\rm hom}_{AC\to B}(x)(\rho_{AC}), \label{eq:combined_instrument_hom} \\
    {\cal F}^{\rm het}_{AC\to B}(\rho_{AC}) &= \iint_{-\infty}^{\infty} d\omega_r d\omega_i\; {\cal K}'^{\,\rm het}_{AC\to B}(x)(\rho_{AC}) ,\label{eq:combined_instrument_het}
\end{align}
with ${\cal K}'^{\,\rm hom}_{AC\to B}(x)$ (resp.~${\cal K}'^{\,\rm het}_{AC\to B}(\omega)$) given by
\begin{align}
    {\cal K}'^{\,\rm hom}_{AC\to B}(x)(\rho_{AC})&\coloneqq \mathrm{Tr}_R\left[{\cal V}_{B;A\to R}^{\rm hom}(x) \circ \bigl(\mathrm{Id}_A\otimes {\cal K}^{\rm hom}_{C\to B}(x)\bigr)(\rho_{AC})\right], \label{eq:K_prime_hom}\\
    {\cal K}'^{\,\rm het}_{AC\to B}(\omega)(\rho_{AC})&\coloneqq \mathrm{Tr}_R\left[{\cal V}_{B;A\to R}^{\rm het}(\omega) \circ \bigl(\mathrm{Id}_A\otimes {\cal K}^{\rm het}_{C\to B}(\omega)\bigr)(\rho_{AC})\right], \label{eq:K_prime_het}
\end{align}
where $\mathrm{Id}$ denotes the identity map.
Note that the idea of acting the isometry $V_{B; A\to R}^{\rm hom}(x)$ or $V_{B; A\to R}^{\rm het}(\omega)$ is closely related to the twisting operation on the shield system \cite{horodecki2006, Renes2007, horodecki2008, horodecki2009, Bourassa2020}.  The difference is that in our case it acts on the system $A$ in a way that is incompatible with the $Z$-basis measurement on $A$.  This is allowed in a security proof based on complementarity since what we need to prove in the virtual protocol is that the outcome of the $Z$-basis measurement on $B$ is secret to Eve when the system $A$ is traced out \cite{complementarity}; i.e., the system $A$ works as a shield system.

We then introduce a virtual protocol that explicitly incorporates the action of ${\cal F}^{\rm hom}_{AC\to B}$ in Eq.~\eqref{eq:combined_instrument_hom} (resp.~${\cal F}^{\rm het}_{AC\to B}$ in Eq.~\eqref{eq:combined_instrument_het}) in Box 2.
\smallskip

\begin{framed}
{\noindent \bf Box 2: Virtual protocol}

\vspace{0.15cm}
\begin{enumerate}
  \setlength{\leftskip}{0.5cm}
  \item[$1'$.] For each of the $N$ rounds, Alice prepares a qubit $A$ and an optical pulse $\tilde{C}$ in a state $\ket{\Psi}_{A\tilde{C}}$ defined in \eqref{eq:prepared_state} and sends the pulse $\tilde{C}$ to Bob.  Bob receives an optical pulse $C$ for each round.
  \item[$2'$.] For the received pulse $C$ in each round, Bob announces a label in the same way as that at Step~2. 
  Alice and Bob do one of the following procedures according to the label.
  \begin{itemize}
    \item[[signal\textrm{]}] Alice and Bob perform the quantum operation on the system $A$ and the received pulse $C$ specified by the map ${\cal F}^{\rm hom}_{AC\to B}$ defined in Eq.~\eqref{eq:combined_instrument_hom} (resp.~${\cal F}^{\rm het}_{AC\to B}$ defined in Eq.~\eqref{eq:combined_instrument_het}) to determine success or failure of detection, obtain the qubit $B$ upon success, and perform the controlled isometry given in Eq.~\eqref{eq:controlled-isometry_1} (resp.~Eq.~\eqref{eq:controlled-isometry_2}).  Bob announces the success or failure of the detection. 
    \item[[test\textrm{]}] Bob performs a heterodyne measurement on the received optical pulse $C$, and obtains an outcome $\hat{\omega}$. 
    Alice measures her qubit $A$ in the $Z$ basis and announces the outcome $\hat{a}\in \{0,1\}$.  Bob calculates the value of $\Lambda_{m,r}(|\hat{\omega} -(-1)^{\hat{a}}\beta|^2)$.
    \item[[trash\textrm{]}] Alice measures her qubit $A$ in the $X$ basis to obtain $\hat{a}'\in \{+,-\}$.
  \end{itemize}
  \item[$3'$.] $\hat{N}^{\rm suc},  \hat{N}^{\rm fail}, \hat{N}^{\rm test}, \hat{N}^{\rm trash}$, and $\hat{F}$ are defined in the same way as those at Step~3. Let $\hat{Q}_{-}$ be the number of rounds with $\hat{a}'=-$ among the $\hat{N}^{\rm trash}$ trash rounds.
  \item[$4'$.]  According to (the upper bound on) the bit error rate $e_{\rm qber}$, Bob performs $H_{\rm EC}$ bits of encrypted communication consuming a pre-shared secret key to send a dummy message.
  \item[$5'$.] Bob computes and announces the final key length $\hat{N}^{\rm fin}$ according to Eq.~\eqref{eq:final_key_length}.  Bob performs a randomly chosen unitary on his qubits (see the main text), and measures the first $\hat{N}^{\rm fin}$ qubits in the $Z$ bases.
\end{enumerate}
\end{framed}

\noindent In the last line of Step $5'$, the random choice of a unitary is constructed so that, along with the subsequent $\hat{N}^{\rm fin}$-qubit measurement in the $Z$ bases, it is equivalent to the privacy amplification.  This is possible because for any $n\times n$ linear transformation $C$ on the $n$-bit sequence, there always exists a corresponding unitary $U(C)$ that satisfies $U(C)\ket{\bm{z}}=\ket{C\bm{z}}$ in the $Z$ basis.  
As has already been claimed, if Eve performs the same attacks as those in the actual protocol, the resulting classical-quantum state between Bob and Eve is the same as that in the actual protocol.  

The complementarity argument \cite{complementarity} in a reverse reconciliation scenario relates the amount of privacy amplification to the so-called phase error patterns of Bob's qubits.  Suppose that, just before the $Z$-basis measurement at Step~$5'$ of the virtual protocol, Bob's quantum state on the first $\hat{N}^{\rm fin}$ qubits is $\varepsilon_{\rm sc}$-close to $\ket{+}\!\bra{+}^{\otimes \hat{N}^{\rm fin}}$ in terms of the infidelity when averaged over $\hat{N}^{\rm fin}$, where the cases $\hat{N}^{\rm fin}=0$ are regarded to have no infidelity.  Then, the secrecy condition of the final key is satisfied with a secrecy parameter $\sqrt{2\varepsilon_{\rm sc}}$ \cite{Hayashi_Tsurumaru, matsuura2019, matsuura2023digital}.  For this to be true, the errors in the $X$ bases (i.e., the phase errors) on Bob's qubits should be corrected by the procedure at Step~$5'$ of the virtual protocol with failure probability no larger than $\varepsilon_{\rm sc}$ when averaged over $\hat{N}^{\rm fin}$ ($\hat{N}^{\rm fin}=0$ cases are regarded as no failure).  To see the correctability of the phase errors at Step~$5'$, suppose that Bob measured his $\hat{N}^{\rm suc}$ qubits in the $X$ basis $\{\ket{+}, \ket{-} \}$ at the end of Step~$3'$, and obtained a sequence of $+$ and $-$.  The minuses in the sequence are regarded as phase errors.
It has already been known that, if we can find an upper bound on the number of possible phase-error patterns, then we can prove the security \cite{complementarity,matsuura2023digital}.  To make the argument more precise, we introduce an estimation protocol in Box 3.

\begin{framed}
{\noindent \bf Box 3: Estimation protocol}

\vspace{0.15cm}
\begin{enumerate}
  \setlength{\leftskip}{0.5cm}
  \item[$1''$.] For each of the $N$ rounds, Alice prepares a qubit $A$ and an optical pulse $\tilde{C}$ in a state $\ket{\Psi}_{A\tilde{C}}$ defined in \eqref{eq:prepared_state} and sends the pulse $\tilde{C}$ to Bob.  Bob receives an optical pulse $C$ for each round.
  \item[$2''$.] For the received pulse $C$ in the $i$th round ($i=1,\ldots, N$), Bob announces a label in the same way as that at Step~2.  Alice and Bob do one of the following procedures according to the label and obtain the values of random variables $\hat{N}^{{\rm suc}\,(i)}_{\rm ph}$, $\hat{F}^{(i)}$, and $\hat{Q}_-^{(i)}$.  Unless explicitly written, these random variables are set to be zeros. 
  \begin{itemize}
    \item[[signal\textrm{]}] Alice and Bob do the same procedure as that at ``signal'' of Step $2'$.  Upon ``success'', Bob performs the $X$-basis measurement on qubit $B$ and obtains $\hat{b}'\in\{+,-\}$.  When $\hat{b}'=-$, $\hat{N}^{{\rm suc}\,(i)}_{\rm ph}$ is set to be unity.
    \item[[test\textrm{]}] Alice and Bob do the same procedure as that at ``test'' of Step $2'$.  Then $\hat{F}^{(i)}$ is set to be $\Lambda_{m,r}(|\hat{\omega} -(-1)^{\hat{a}}\beta|^2)$.
    \item[[trash\textrm{]}] Alice does the same procedure as that at ``trash'' of Step $2'$.  When $\hat{a}'=-$, $\hat{Q}_-^{(i)}$ is set to be unity.   
  \end{itemize}
  \item[$3''$.] Same as Steps $3'$ of the virtual protocol.  Note that $\hat{F}=\sum_{i=1}^{N}\hat{F}^{(i)}$ and $\hat{Q}_-=\sum_{i=1}^{N}\hat{Q}_-^{(i)}$ hold.
  \item[$4''$.]  Regarding $+$ as zero and $-$ as unity for each $\hat{b}'$ in success signal round, define the $\hat{N}^{\rm suc}$-bit sequence $\hat{\bm{x}}_{\rm ph}$.  Let $\hat{N}^{\rm suc}_{\rm ph}$ be the Hamming weight of $\hat{\bm{x}}_{\rm ph}$, i.e., $\hat{N}^{\rm suc}_{\rm ph}=\sum_{i=1}^{N}\hat{N}^{{\rm suc}\,(i)}_{\rm ph}$.
  \item[$5''$.] Bob performs a universal${}_2$ hashing on $\hat{\bm{x}}_{\rm ph}$ and obtains $(\hat{N}^{\rm suc} - \hat{N}^{\rm fin})$-bit syndrome.  From this syndrome, Bob estimates the binary string $\hat{\bm{x}}_{\rm ph}$.
\end{enumerate}
\end{framed}

\noindent If we could uniquely identify the $\hat{N}^{\rm suc}$-bit sequence $\hat{\bm{x}}_{\rm ph}$ with failure probability at most $\varepsilon_{\rm sc}$ (when averaged over $\hat{N}^{\rm fin}$ with $\hat{N}^{\rm fin}=0$ cases regarded as no failure), then we could identify the phase errors in the virtual protocol as well and correct them by appropriately acting Pauli-$Z$ operations, and thus the actual protocol can be made $\sqrt{2\varepsilon_{\rm sc}}$-secret.  Therefore, the task of proving the security of the actual protocol is now reduced to identifying the $\hat{N}^{\rm suc}$-bit sequence $\hat{\bm{x}}_{\rm ph}$ with a bounded failure probability (when averaged over $\hat{N}^{\rm fin}$) in the estimation protocol.

We show in the following that the task is further reduced to constructing a function $U(\hat{F}, \hat{N}^{\rm trash})$ that satisfies 
\begin{equation}
  {\rm Pr}\left[\hat{N}_{\rm ph}^{\rm suc} > U(\hat{F}, \hat{N}^{\rm trash}), \hat{N}^{\rm suc}\geq 1\right]\leq \epsilon \label{eq:probability_condition}
\end{equation}
for any attack in the estimation protocol, and setting the final key length to $\hat{N}^{\rm fin}=\max\{\hat{N}^{\rm suc}-H_{\rm PA}-s, 0\}$, where $H_{\rm PA}$ is defined as
\begin{equation}
    H_{\rm PA}\coloneqq \Bigl\lceil \hat{N}^{\rm suc} h\!\(\min\bigl\{U(\hat{F}, \hat{N}^{\rm trash})/\hat{N}^{\rm suc}, 0.5\bigr\}\) \Bigr\rceil.
\end{equation}
In fact, if the condition~\eqref{eq:probability_condition} is satisfied, then the number of possible patterns of $\hat{\bm{x}}_{\rm ph}$ can be bounded from above by $2^{H_{\rm PA}}$ \cite{Cover2012,matsuura2023digital} except for the failure case with its probability no larger than $\epsilon$.
When the number of candidates of $\hat{\bm{x}}_{\rm ph}$ were restricted to $2^{H_{\rm PA}}$, Bob could uniquely identify $\hat{\bm{x}}_{\rm ph}$ with failure probability no larger than $2^{-s}$ by extracting an $(H_{\rm PA}+s)$-bit syndrome from $\hat{\bm{x}}_{\rm ph}$ using a universal${}_2$ hash function \cite{complementarity,tsurumaru2013,tsurumaru2020leftover,tsurumaru2020equivalence,matsuura2023digital}.
Therefore, using the union bound, the failure probability of identifying $\hat{\bm{x}}_{\rm ph}$ can be bounded from above by $\epsilon+2^{-s}$ (when averaged over $\hat{N}^{\rm fin}$) if the condition~\eqref{eq:probability_condition} is satisfied and the final key length is set to be $\max\{\hat{N}^{\rm suc}-H_{\rm PA}-s,0\}$ at the end of the estimation protocol.

Finally, we mention that Step~$5'$ in the virtual protocol realizes a phase error correction.
In fact, the $(\hat{N}^{\rm suc}-\hat{N}^{\rm fin})$-bit syndrome extraction in the $X$ bases via the universal${}_2$ hash function can be realized through the randomly chosen unitary at Step~$5'$ in the virtual protocol, and the error recovery in the $X$ bases of the rest $\hat{N}^{\rm fin}$ qubits by the Pauli-$Z$ action leaves $Z$-basis values of these qubits unchanged and thus irrelevant (due to the successive $Z$-basis measurement at Step~$5'$).  Since a unitary $U(C^{-1})$ that acts as the matrix $C^{-1}$ in the $Z$ bases acts as $C^{\top}$ in the $X$ bases, i.e., $U(C^{-1})\ket{\bm{x}_X}=\ket{C^{\top}\bm{x}_X}$ where ${\cdot}_X$ denotes the $X$ basis, the (random) unitary that acts as the universal${}_2$ hashing in the $X$ bases of the last $(\hat{N}^{\rm suc}-\hat{N}^{\rm fin})$ qubits acts as the dual universal${}_2$ hashing in the $Z$ bases of the first $\hat{N}^{\rm fin}$ qubits \cite{tsurumaru2013,tsurumaru2020equivalence} (i.e., in the actual protocol).  Thus, we can conclude that the actual protocol is $\sqrt{2}\sqrt{\epsilon+2^{-s}}$-secret once the condition~\eqref{eq:probability_condition} is satisfied in the estimation protocol, the final key length is set as Eq.~\eqref{eq:final_key_length}, and the privacy amplification is done by the dual universal${}_2$ hashing.

Combining these, the condition \eqref{eq:probability_condition} implies that the actual protocol with the adaptive final key length given in Eq.~\eqref{eq:final_key_length} is $\varepsilon_{\rm sec}$-secure with a security parameter
$\varepsilon_{\rm sec}=\sqrt{2}\sqrt{\epsilon+2^{-s}}+\varepsilon_{\rm cor}$ \cite{complementarity,Hayashi_Tsurumaru,matsuura2019}.  From now on, we thus focus on the estimation protocol for finding a function $U(\hat{F}, \hat{N}^{\rm trash})$ to satisfy Eq.~\eqref{eq:probability_condition}.

\subsection{Phase error operator} \label{sec:phase_error_operator}
In this section, we explain how to obtain refined phase error operators that eventually lead to tighter bounds on phase errors than those in Refs.~\cite{matsuura2021, Yamano2022}.
The phase error operators depend on the choice of the controlled isometry $V_{B; A\to R}^{\rm hom}(x)$ or $V_{B; A\to R}^{\rm het}(\omega)$ in the virtual and the estimation protocol.
As we mentioned previously, when the isometries $V^{(0)}_{A\to R}(x)$ and $V^{(1)}_{A\to R}(x)$  in the controlled isometry $V_{B; A\to R}^{\rm hom}(x)$ defined in Eq.~\eqref{eq:controlled-isometry_1} satisfy $V^{(0)}_{A\to R}(x)=V^{(1)}_{A\to R}(x)$, the analysis reduces to the previous one \cite{matsuura2021}.  This is because in this case the phase error can be defined as the state $\ket{-}$ in the $X$ basis of Bob's qubit after acting Controlled-NOT operation from Bob's qubit to Alice's, as can be seen from Eqs.~\eqref{eq:controlled-isometry_1}, \eqref{eq:combined_instrument_hom}, and \eqref{eq:K_prime_hom}.  This is equivalent to observing the discrepancy between the $X$-basis values of Alice's and Bob's qubits before the action of Controlled-NOT, and thus equivalent to the definition of the phase error in Ref.~\cite{matsuura2021}.  The relation between our new analysis for Heterodyne protocol and that in Ref.~\cite{Yamano2022} can be interpreted in the same way.
Here in our proposal, we have a degree of freedom to ``twist'' the Alice's qubit system with these isometries depending on protocol parameters as well as a homodyne (resp.~heterodyne) outcome.  Unfortunately, finding the optimal choice of these isometries with respect to all protocol parameters and observables in the actual protocol may be intractable.
We thus take a suboptimal strategy; fix $V_{B; A\to R}^{\rm hom}(x)$ (resp.~$V_{B; A\to R}^{\rm het}(\omega)$) so that the probability of the phase error event $\hat{b}'=-$ in the estimation protocol is minimized for an ideal pure-loss channel \cite{pirandola2017} with transmission $\eta=\beta^2/\mu$, where $\beta$ and $\mu$ are the positive parameters predetermined in the actual protocol.
There are several reasons to do this.  One is that the loss is the dominant noise for CV-QKD protocols and thus we prioritize the robustness against it to obtain better performance.  Another is that the optimization for the pure-loss channel is much easier since we have an analytical expression of the state when the initial state Eq.~\eqref{eq:prepared_state} in the virtual protocol is subject to the pure-loss channel.
In fact, when the state $\ket{\Psi}_{A\tilde{C}}$ in Eq.~\eqref{eq:prepared_state} is put into a pure-loss channel with the channel output being $\ket{\pm\beta}_C$, the resulting state $\ket{\Phi}_{ACE}$ on systems $A, C$, and an adversary's system $E$ (i.e., an environment of the pure-loss channel) is given by
\begin{equation}
    \ket{\Phi}_{ACE} = \frac{1}{\sqrt{2}}\(\ket{0}_A \Ket{\beta}_C \Ket{\sqrt{\mu - \beta^2}}_E + \ket{1}_A \Ket{-\beta}_C \Ket{-\sqrt{\mu - \beta^2}}_E\).
    \label{eq:initial_state}
\end{equation} 
Tracing out the system $E$, the reduced state $\Phi_{AC}$ is given by
\begin{equation}
    \Phi_{AC} = (1 - q_{\mu,\beta}) \ket{\phi_+}\!\bra{\phi_+}_{AC} + q_{\mu,\beta} \ket{\phi_-}\!\bra{\phi_-}_{AC},
    \label{eq:pure_loss_output}
\end{equation}
where 
\begin{align}
    \ket{\phi_+}_{AC} &\coloneqq \frac{1}{\sqrt{2}}(\ket{0}\ket{\beta} + \ket{1}\ket{-\beta}) = \ket{+}_A\otimes\Pi_{\rm ev}\ket{\beta}_C + \ket{-}_A\otimes\Pi_{\rm od}\ket{\beta}_C, \label{eq:phi_plus}\\
    \ket{\phi_-}_{AC} &\coloneqq \frac{1}{\sqrt{2}}(\ket{0}\ket{\beta} - \ket{1}\ket{-\beta}) = \ket{+}_A\otimes\Pi_{\rm od}\ket{\beta}_C + \ket{-}_A\otimes\Pi_{\rm ev}\ket{\beta}_C = (Z_A\otimes I_C) \ket{\phi_+}_{AC}, \label{eq:phi_minus}
\end{align}
and
\begin{equation}
    q_{\mu,\beta} \coloneqq \frac{1 - e^{-2(\mu-\beta^2)}}{2} (>0).
\end{equation}

Now we aim at obtaining the optimal choice of $V_{B; A\to R}^{\rm hom}(x)$ (resp.~$V_{B; A\to R}^{\rm het}(\omega)$) so that the probability that Bob obtains $\hat{b}'=-$ is minimized for the state $\Phi_{AC}$ in Eq.~\eqref{eq:pure_loss_output}.
For Homodyne protocol, we observe that
\begin{align}
    &\text{C-}X_{BA}\,\bigl({\rm Id}_A\otimes\mathcal{K}_{C\to B}^{\rm hom}(x)\bigr)(\Phi_{AC})\, \text{C-}X_{BA} \label{eq:tau_hom_def}\\
    \begin{split}
    &= 2f_{\rm suc}(x)\,\text{C-}X_{BA} \left[(1 - q_{\mu,\beta})\,\hat{P}\bigl(\bra{x}\Pi_{\rm ev}\ket{\beta}\ket{++}_{AB} + \bra{x}\Pi_{\rm od}\ket{\beta}\ket{--}_{AB}\bigr) \right.\\
    &\hspace{3.5cm} \left. + q_{\mu,\beta}\,\hat{P}\bigl(\bra{x}\Pi_{\rm od}\ket{\beta}\ket{+-}_{AB} + \bra{x}\Pi_{\rm ev}\ket{\beta}\ket{-+}_{AB}\bigr)\right]\text{C-}X_{BA}
    \end{split}\\
    \begin{split}
    &= 2f_{\rm suc}(x) \left[(1 - q_{\mu,\beta})\,\hat{P}\bigl(\bra{x}\Pi_{\rm ev}\ket{\beta}\ket{+}_{A} + \bra{x}\Pi_{\rm od}\ket{\beta}\ket{-}_{A}\bigr)\otimes\ket{+}\!\bra{+}_B \right.\\
    &\hspace{3.5cm} \left.+ q_{\mu,\beta}\,\hat{P}\bigl(\bra{x}\Pi_{\rm od}\ket{\beta}\ket{+}_{A} + \bra{x}\Pi_{\rm ev}\ket{\beta}\ket{-}_{A}\bigr)\otimes\ket{-}\!\bra{-}_B \right]
    \end{split}\\
    \begin{split}
    &= f_{\rm suc}(x) \left[(1 - q_{\mu,\beta})\,\hat{P}\Bigl(\sqrt{g_{\beta,1/4}(x)}\ket{0}_{A} + \sqrt{g_{-\beta,1/4}(x)}\ket{1}_{A}\Bigr)\otimes\ket{+}\!\bra{+}_B  \right. \\
    &\hspace{3.5cm} \left. + q_{\mu,\beta}\,\hat{P}\Bigl(\sqrt{g_{\beta,1/4}(x)}\ket{0}_{A} - \sqrt{g_{-\beta,1/4}(x)}\ket{1}_{A}\Bigr)\otimes\ket{-}\!\bra{-}_B\right],\label{eq:phase_entangled}
    \end{split} 
\end{align}
where $\hat{P}(\psi)\coloneqq \psi \psi^{\dagger}$ (and thus $\hat{P}(\ket{\psi})=\ket{\psi}\!\bra{\psi}$), and $g_{m,V}$ is the normal distribution with the mean $m$ and the variance $V$, i.e.,
\begin{equation}
    g_{m,V}(x) \coloneqq \frac{1}{\sqrt{2\pi V}} \exp\left[-\frac{(x-m)^2}{2V}\right].
\end{equation}
In the above (i.e., Eqs.~\eqref{eq:tau_hom_def}--\eqref{eq:phase_entangled}), the first equality follows from Eq.~\eqref{eq:K_x}, and the third follows from Eq.~\eqref{eq:wave_func_coherent} as well as the fact that $\beta$ is real.
We define $\tau_{AB}^{\rm hom}(x)$ as
\begin{equation}
    \begin{split}
    \tau_{AB}^{\rm hom}(x) \coloneqq 
    &(1 - q_{\mu,\beta})\,\hat{P}\Bigl(\sqrt{g_{\beta,1/4}(x)}\ket{0}_{A} + \sqrt{g_{-\beta,1/4}(x)}\ket{1}_{A}\Bigr)\otimes\ket{+}\!\bra{+}_B   \\
    &\hspace{2.5cm}+ q_{\mu,\beta}\,\hat{P}\Bigl(\sqrt{g_{\beta,1/4}(x)}\ket{0}_{A} - \sqrt{g_{-\beta,1/4}(x)}\ket{1}_{A}\Bigr)\otimes\ket{-}\!\bra{-}_B.
    \end{split} \label{eq:tau_hom}
\end{equation}
From Eqs.~\eqref{eq:controlled-isometry_1}, \eqref{eq:K_prime_hom}, \eqref{eq:phase_entangled}, and \eqref{eq:tau_hom}, the probability density of an outcome $x$ with occurrence of the phase error is given by
\begin{align}
    &\mathrm{Tr}\Bigl[\ket{-}\!\bra{-}_B \,{\cal K}'^{\,\rm hom}_{AC\to B}(x)(\Phi_{AC})\Bigr] \nonumber\\
    \begin{split}&= \frac{f_{\rm suc}(x)}{2}  \,\mathrm{Tr}\Bigl[\bra{0}_B\tau_{AB}^{\rm hom}(x)\ket{0}_B + \bra{1}_B\tau_{AB}^{\rm hom}(x)\ket{1}_B \label{eq:phase_error_prob}\\
    & \qquad  - \left(V_{A\to R}^{(1)}(x)\right)^{\dagger}V_{A\to R}^{(0)}(x) \bra{0}_B\tau_{AB}^{\rm hom}(x)\ket{1}_B -  \bra{1}_B\tau_{AB}^{\rm hom}(x)\ket{0}_B \left(V_{A\to R}^{(0)}(x)\right)^{\dagger} V_{A\to R}^{(1)}(x)\Bigr]
    \end{split}\\
    &= f_{\rm suc}(x)\left[\frac{1}{2}\; \mathrm{Tr}\left(\tau_{AB}^{\rm hom}(x) \right) - \mathrm{Re}\left(\mathrm{Tr}\left[ \left(V_{A\to R}^{(1)}(x)\right)^{\dagger} V_{A\to R}^{(0)}(x) \bra{0}_B\tau_{AB}^{\rm hom}(x)\ket{1}_B \right]\right)\right] \\
    &\geq f_{\rm suc}(x)\left[\frac{1}{2}\; \mathrm{Tr}\left(\tau_{AB}^{\rm hom}(x) \right) - \left\| \bra{0}_B\tau_{AB}^{\rm hom}(x)\ket{1}_B \right\|_1\right], \label{eq:minimum_phase_err_prob}
\end{align} 
where the last inequality follows from the matrix H\"{o}lder inequality.  If we write the polar decomposition of $\bra{0}_B\tau_{AB}^{\rm hom}(x)\ket{1}_B$ as $W_{A}^{\rm hom}(x) \bigl| \bra{0}_B\tau_{AB}^{\rm hom}(x)\ket{1}_B \bigr|$, the equality in \eqref{eq:minimum_phase_err_prob} can be achieved by setting 
\begin{equation}
    \left(V_{A\to R}^{(1)}(x)\right)^{\dagger}V_{A\to R}^{(0)}=\left(W_{A}^{{\rm hom}}(x)\right)^{\dagger}. \label{eq:Vs_to_W_hom}
\end{equation}
From Eq.~\eqref{eq:tau_hom}, $\bra{0}_B\tau_{AB}^{\rm hom}(x)\ket{1}_B$ is given by
\begin{equation}
    \begin{split}
    \bra{0}_B\tau_{AB}^{\rm hom}(x)\ket{1}_B &= \frac{1}{2}\left[(1 - q_{\mu,\beta})\,\hat{P}\Bigl(\sqrt{g_{\beta,1/4}(x)}\ket{0}_{A} + \sqrt{g_{-\beta,1/4}(x)}\ket{1}_{A}\Bigr) \right.\\ 
    &\hspace{3cm}\left. - q_{\mu,\beta}\,\hat{P}\Bigl(\sqrt{g_{\beta,1/4}(x)}\ket{0}_{A} - \sqrt{g_{-\beta,1/4}(x)}\ket{1}_{A}\Bigr)\right],\label{eq:tau_hom_0_1}
    \end{split}
\end{equation}
which is hermitian with two eigenvalues having opposite signs.  Let $\ket{u_+^{\rm hom}(x)}_A$ and $\ket{u_-^{\rm hom}(x)}_A$ be eigenvectors of $\bra{0}_B\tau_{AB}^{\rm hom}(x)\ket{1}_B$ with positive and negative eigenvalues, respectively.  Then, $W_A^{\rm hom}(x)$ is given by
\begin{equation}
    W_A^{\rm hom}(x) = \ket{u_+^{\rm hom}(x)}\!\bra{u_+^{\rm hom}(x)}_A - \ket{u_-^{\rm hom}(x)}\!\bra{u_-^{\rm hom}(x)}_A.\label{eq:W_A_hom}
\end{equation}
The explicit form of $\ket{u_\pm^{\rm hom}(x)}_A$ is given in Eq.~\eqref{eq:form_of_u_pm} in Appendix~\ref{sec:operator_ineq}.
The choice of the isometry $V_{A\to R}^{(j)}(x)$ to satisfy Eq.~\eqref{eq:Vs_to_W_hom} is not unique; one of the reasons is the arbitrariness of the dimension of the system $R$. 
Here, we set $R=A$ and set
\begin{equation}
    \begin{split}
        V_{A\to R}^{(0)}(x) &= I_A, \\
        V_{A\to R}^{(1)}(x) &= W_A^{\rm hom}(x) ,
    \end{split}\label{eq:V_i_hom}
\end{equation}
which, with Eqs.~\eqref{eq:controlled-isometry_1} and \eqref{eq:W_A_hom}, leads to
\begin{equation}
    V^{\rm hom}_{B;A\to A}(x) = \left[\ket{u_+^{\rm hom}(x)}\!\bra{u_+^{\rm hom}(x)}_A \otimes I_B + \ket{u_-^{\rm hom}(x)}\!\bra{u_-^{\rm hom}(x)}_A\otimes Z_B\right] \text{C-}X_{BA} .\label{eq:controlled_iso_hom_alt}
\end{equation}

For Heterodyne protocol, the calculation similar to Eqs.~\eqref{eq:tau_hom_def}--\eqref{eq:tau_hom} leads to
\begin{align} 
    &\text{C-}X_{BA}\,\bigl({\rm Id}_A\otimes\mathcal{K}_{C\to B}^{\rm het}(\omega)\bigr)(\Phi_{AC})\, \text{C-}X_{BA}\\
    \begin{split}
        &= \frac{2f_{\rm suc}(\omega_r)}{\pi}\,\text{C-}X_{BA} \left[(1 - q_{\mu,\beta})\,\hat{P}\bigl(\bra{\omega}\Pi_{\rm ev}\ket{\beta}\ket{++}_{AB} + \bra{\omega}\Pi_{\rm od}\ket{\beta}\ket{--}_{AB}\bigr) \right.\\
        &\hspace{4cm} \left. + q_{\mu,\beta}\,\hat{P}\bigl(\bra{\omega}\Pi_{\rm od}\ket{\beta}\ket{+-}_{AB} + \bra{\omega}\Pi_{\rm ev}\ket{\beta}\ket{-+}_{AB}\bigr)\right]\text{C-}X_{BA}
    \end{split}\\
    \begin{split}
        &= \frac{f_{\rm suc}(\omega_r)}{\pi} \left[(1 - q_{\mu,\beta})\,\hat{P}\bigl(\braket{\omega|\beta}\ket{0}_{A} + \braket{-\omega|\beta}\ket{1}_{A}\bigr)\otimes\ket{+}\!\bra{+}_B  \right. \\
        &\hspace{4cm} \left. + q_{\mu,\beta}\,\hat{P}\bigl(\braket{\omega|\beta}\ket{0}_{A} - \braket{-\omega|\beta}\ket{1}_{A}\bigr)\otimes\ket{-}\!\bra{-}_B\right],
    \end{split} 
\end{align}
Since $\braket{\omega|\beta} = e^{-\frac{1}{2}[(\omega_r - \beta)^2 + \omega_i^2 + 2i\omega_i\beta]}$ is not real in general, we heuristically insert a $\theta$-rotation around the $Z$ basis
\begin{equation}
    R^Z_{A}(\theta)\coloneqq \exp(-i\theta Z_A/2) \label{eq:rotation_around_z}
\end{equation} 
in order to have
\begin{align}
        &\left(R^Z_{A}(2\omega_i\beta)\right)^{\dagger} \text{C-}X_{BA}\,\bigl({\rm Id}_A\otimes\mathcal{K}_{C\to B}^{\rm het}(\omega)\bigr)(\Phi_{AC})\, \text{C-}X_{BA} \,R^{Z}_{A}(2\omega_i\beta) \\
    \begin{split}
        &= \frac{e^{-\omega_i^2}f_{\rm suc}(\omega_r)}{\sqrt{\pi}} \left[(1 - q_{\mu,\beta})\,\hat{P}\Bigl(\sqrt{g_{\beta,1/2}(\omega_r)}\ket{0}_{A} + \sqrt{g_{-\beta,1/2}(\omega_r)}\ket{1}_{A}\Bigr)\otimes\ket{+}\!\bra{+}_B  \right. \\
        &\hspace{3.5cm} \left. + q_{\mu,\beta}\,\hat{P}\Bigl(\sqrt{g_{\beta,1/2}(\omega_r)}\ket{0}_{A} - \sqrt{g_{-\beta,1/2}(\omega_r)}\ket{1}_{A}\Bigr)\otimes\ket{-}\!\bra{-}_B\right]. 
    \end{split}
\end{align}
We define $\tau_{AB}^{\rm het}(\omega_r)$ as 
\begin{equation}
    \begin{split}
    \tau_{AB}^{\rm het}(\omega_r) &\coloneqq (1 - q_{\mu,\beta})\,\hat{P}\Bigl(\sqrt{g_{\beta,1/2}(\omega_r)}\ket{0}_{A} + \sqrt{g_{-\beta,1/2}(\omega_r)}\ket{1}_{A}\Bigr)\otimes\ket{+}\!\bra{+}_B  \\
    &\hspace{2.5cm} + q_{\mu,\beta}\,\hat{P}\Bigl(\sqrt{g_{\beta,1/2}(\omega_r)}\ket{0}_{A} - \sqrt{g_{-\beta,1/2}(\omega_r)}\ket{1}_{A}\Bigr)\otimes\ket{-}\!\bra{-}_B.
    \end{split}\label{eq:tau_het}
\end{equation}
Thus, the structure of the matrix $\tau_{AB}^{\rm het}(\omega_r)$ is essentially the same as $\tau_{AB}^{\rm hom}(x)$ of Homodyne protocol.  (The heuristics of inserting $R^Z_{A}(\theta)$ above is for this reduction.)
In the same way as Homodyne protocol, the probability density of outcome $\omega$ with the occurrence of a phase error is given by
\begin{align}
    &\mathrm{Tr}\Bigl[\ket{-}\!\bra{-}_B \,{\cal K}'^{\,\rm het}_{AC\to B}(\omega)(\Phi_{AC})\Bigr]\nonumber\\
    \begin{split}
    &= \frac{e^{-\omega_i^2}f_{\rm suc}(\omega_r)}{\sqrt{\pi}}\left[\frac{1}{2}\;\mathrm{Tr}\left(\tau_{AB}^{\rm het}(\omega_r)\right) \right.\\
    &\hspace{3cm}\left. - \mathrm{Re}\left(\mathrm{Tr}\left[ \bigl(V'^{(1)}_{A\to R}(\omega_r)\bigr)^{\dagger} V'^{(0)}_{A\to R}(\omega_r) R^Z_{A}(2\omega_i\beta) \bra{0}_B\tau_{AB}^{\rm het}(\omega_r)\ket{1}_B \left(R^Z_{A}(2\omega_i\beta)\right)^{\dagger}  \right]\right) \right] 
    \end{split}\\
    & \geq \frac{e^{-\omega_i^2}f_{\rm suc}(\omega_r)}{\sqrt{\pi}} \left[\frac{1}{2}\;\mathrm{Tr}\left(\tau_{AB}^{\rm het}(\omega_r)\right) - \left\| \bra{0}_B\tau_{AB}^{\rm het}(\omega_r)\,\ket{1}_B  \right\|_1 \right]. \label{eq:minimum_phase_err_prob_het}
\end{align}
If we write the polar decomposition of $\bra{0}_B\tau_{AB}^{\rm het}(\omega_r)\,\ket{1}_B$ by $W_{A}^{\rm het}(\omega_r) \bigl|\bra{0}_B\tau_{AB}^{\rm het}(\omega_r)\,\ket{1}_B\bigr|$, then the equality of Eq.~\eqref{eq:minimum_phase_err_prob_het} can be achieved by setting 
\begin{equation}
    \left(R^Z_{A}(2\omega_i\beta)\right)^{\dagger}\left(V'^{(1)}_{A\to R}(\omega)\right)^{\dagger} V'^{(0)}_{A\to R}(\omega)R^Z_{A}(2\omega_i\beta) =\left(W_{A}^{\rm het}(\omega_r)\right)^{\dagger}. \label{eq:Vs_to_W_het}
\end{equation}
From Eq.~\eqref{eq:tau_het}, $\bra{0}_B\tau_{AB}^{\rm het}(\omega_r)\,\ket{1}_B$ is given by
\begin{equation}
    \begin{split}
        \bra{0}_B \tau_{AB}^{\rm het}(\omega_r)\,\ket{1}_B 
        &= \frac{1}{2} \left[(1 - q_{\mu,\beta})\,\hat{P}\Bigl(\sqrt{g_{\beta,1/2}(\omega_r)}\ket{0}_{A} + \sqrt{g_{-\beta,1/2}(\omega_r)}\ket{1}_{A}\Bigr) \right. \\
        &\hspace{3cm} \left. - q_{\mu,\beta}\,\hat{P}\Bigl(\sqrt{g_{\beta,1/2}(\omega_r)}\ket{0}_{A} - \sqrt{g_{-\beta,1/2}(\omega_r)}\ket{1}_{A}\Bigr)\right], \label{eq:tau_het_0_1}
    \end{split}
\end{equation}
which is hermitian.
Let $\ket{u_+^{\rm het}(\omega_r)}_A$ and $\ket{u_-^{\rm het}(\omega_r)}_A$ be eigenvectors of $\bra{0}_B \tau_{AB}^{\rm het}(\omega_r)\,\ket{1}_B $ with positive and negative eigenvalues, respectively.  Then, $W_A^{\rm het}(\omega_r)$ is given by 
\begin{equation}
    W_A^{\rm het}(\omega_r) = \ket{u_+^{\rm het}(\omega_r)}\!\bra{u_+^{\rm het}(\omega_r)}_A - \ket{u_-^{\rm het}(\omega_r)}\!\bra{u_-^{\rm het}(\omega_r)}_A. \label{eq:W_A_het}
\end{equation}
We can choose $V'^{(j)}_{A\to R}(\omega)$ to satisfy Eq.~\eqref{eq:Vs_to_W_het} in the same way as Homodyne protocol.  We set $R=A$ and set
\begin{equation}
    \begin{split}
        V'^{(0)}_{A\to R}(\omega) &= \left(R^Z_{A}(2\omega_i\beta)\right)^{\dagger}, \\
        V'^{(1)}_{A\to R}(\omega) &= W_A^{\rm het}(\omega_r)\left(R^Z_{A}(2\omega_i\beta)\right)^{\dagger},
    \end{split} \label{eq:V_i_het}
\end{equation}
which, with Eqs.~\eqref{eq:controlled-isometry_2} and \eqref{eq:W_A_het}, leads to
\begin{equation}
    V^{\rm het}_{B;A\to A}(\omega) = \left[\ket{u_+^{\rm het}(\omega_r)}\!\bra{u_+^{\rm het}(\omega_r)}_A \otimes I_B + \ket{u_-^{\rm het}(\omega_r)}\!\bra{u_-^{\rm het}(\omega_r)}_A\otimes Z_B\right]\left(R^Z_{A}(2\omega_i\beta)\right)^{\dagger} \text{C-}X_{BA} .\label{eq:controlled_iso_het_alt}
\end{equation}

We thus obtained the optimal choice of the controlled isometry $V_{B; A\to R}^{\rm hom}(x)$ in Eq.~\eqref{eq:controlled_iso_hom_alt} for Homodyne protocol (resp.~$V_{B; A\to R}^{\rm het}(\omega)$ in Eq.~\eqref{eq:controlled_iso_het_alt} for Heterodyne protocol) so that the probability that Bob obtains $\hat{b}'=-$ is minimized for the pure-loss channel.
As explained previously, we set $V_{A\to R}^{(j)}(x)$ to the one in Eq.~\eqref{eq:V_i_hom} (resp.~$V'^{(j)}_{A\to R}(\omega)$ to the one in Eq.~\eqref{eq:V_i_het}) also for arbitrary channels, i.e., arbitrary coherent attacks by Eve.  This choice is suboptimal for general channels but is expected to be close to optimal for channels that are close to the pure-loss one.  
Now that the controlled isometry $V_{B; A\to A}^{\rm hom}(x)$ (resp.~$V_{B; A\to A}^{\rm het}(\omega)$) is fixed, we can interpret the event that Bob announces ``success'' and obtains $\hat{b}'=-$ (i.e., the phase error) at the signal round of Estimation protocol as the outcome of a generalized measurement on Alice's qubit $A$ and the optical pulse $C$, and define the corresponding POVM element $M_{\rm ph}^{\rm hom /het}$ (i.e., the phase error operator) through Eq.~\eqref{eq:combined_instrument_hom} (resp.~Eq.~\eqref{eq:combined_instrument_het}) as
\begin{align}
    M^{\rm hom}_{\rm ph} &\coloneqq {\cal F}^{{\rm hom}\; \ddagger}_{AC\to B}\bigl(\ket{-}\!\bra{-}_{B}\bigr)= \int_{-\infty}^{\infty} dx\, \left({\cal K}'^{\,\rm hom}_{AC\to B}(x)\right)^{\ddagger} \bigl(\ket{-}\!\bra{-}_{B}\bigr),\\
    M^{\rm het}_{\rm ph} &\coloneqq {\cal F}^{{\rm het} \; \ddagger}_{AC\to B}\bigl(\ket{-}\!\bra{-}_{B}\bigr) = \iint_{-\infty}^{\infty} d\omega_r\,d\omega_i\, \left( {\cal K}'^{\,\rm het}_{AC\to B}(\omega)\right)^{\ddagger} \bigl(\ket{-}\!\bra{-}_{B}\bigr),
\end{align}
where $\ddagger$ denotes the adjoint map.  Then, for any density operator $\rho$ on the joint system $AC$, $M_{\rm ph}^{\rm hom}$ (resp.~$M_{\rm ph}^{\rm het}$) satisfies
\begin{equation}
    \mathbb{E}_{\rho}\left[\hat{N}^{{\rm suc}\,(i)}_{\rm ph}\right] = p_{\rm sig}\mathrm{Tr}\!\left[\rho\, M_{\rm ph}^{\rm hom /het}\right] \label{eq:N_suc_expectation}
\end{equation}
in Homodyne (resp.~Heterodyne) protocol, where $\hat{N}_{\rm ph}^{{\rm suc}\,(i)}$ is defined in Estimation protocol (see Box 3).
For Homodyne protocol, by using Eqs.~\eqref{eq:I_homo}, \eqref{eq:K_prime_hom}, and \eqref{eq:controlled_iso_hom_alt}, we have
\begin{align}
    M^{\rm hom}_{\rm ph} 
    &=  \int_{-\infty}^{\infty} dx\, \left[I_A\otimes\bigl(K^{\rm hom}_{\rm suc}(x)\bigr)^{\dagger}\right] \bigl(V^{\rm hom}_{B;A\to A}(x)\bigr)^{\dagger} \bigl(I_A\otimes \ket{-}\!\bra{-}_B\bigr)V^{\rm hom}_{B;A\to A}(x) \left[I_A\otimes K^{\rm hom}_{\rm suc}(x)\right]\\
    \begin{split}
    &=  \int_{-\infty}^{\infty} dx    \left[\hat{P}\!\left(\bigl[I_A\otimes \bigl(K^{\rm hom}_{\rm suc}(x)\bigr)^{\dagger}\bigr]\, \text{C-}X_{BA} \ket{u_+^{\rm hom}(x)}_A\otimes \ket{-}_B  \right)\right.\\
    & \hspace{3cm} \left.+ \hat{P}\left(\bigl[I_A\otimes \bigl(K^{\rm hom}_{\rm suc}(x)\bigr)^{\dagger}\bigr]\, \text{C-}X_{BA} \ket{u_-^{\rm hom}(x)}_A\otimes \ket{+}_B \right)\right],
    \end{split}\label{eq:phase_error_hom}
\end{align}
where we used the fact that the adjoint map of the tracing-out $\mathrm{Tr}_{A}$ is taking the tensor product with $I_A$.  Using the relation $\text{C-}X_{BA} = \ket{+}\!\bra{+}_A\otimes I_B + \ket{-}\!\bra{-}_A \otimes Z_B$ as well as Eq.~\eqref{eq:K_x}, we have
\begin{equation}
    \begin{split}
        M^{\rm hom}_{\rm ph} &= \int_{-\infty}^{\infty} 2f_{\rm suc}(x) dx  \left[\hat{P}\!\left(\Pi^{(+,{\rm od}),(-,{\rm ev})}_{AC}\ket{u_+^{\rm hom}(x)}_A\otimes\ket{x}_C \right) \right. \\
        &\hspace{4cm} \left.  +  \hat{P}\!\left(\Pi^{(-,{\rm od}),(+,{\rm ev})}_{AC}\ket{u_-^{\rm hom}(x)}_A\otimes\ket{x}_C \right) \right],
    \end{split}
\label{eq:hom_direct_sum}
\end{equation}
where two orthogonal projections $\Pi_{AC}^{(+,\mathrm{od}),(-,\mathrm{ev})}$ and $\Pi_{AC}^{(-,\mathrm{od}), (+,\mathrm{ev})}$ are defined as
\begin{align}
    \Pi_{AC}^{(+,\mathrm{od}), (-,\mathrm{ev})} & \coloneqq \ket{+}\!\bra{+}_A \otimes \Pi_{\rm od}+\ket{-}\!\bra{-}_A\otimes \Pi_{\rm ev}, \label{eq:plus_od_minus_ev}\\
    \Pi_{AC}^{(-,\mathrm{od}), (+,\mathrm{ev})} & \coloneqq \ket{-}\!\bra{-}_A\otimes \Pi_{\rm od} + \ket{+}\!\bra{+}_A\otimes \Pi_{\rm ev}.\label{eq:minus_od_plus_ev}
\end{align}
A similar relation holds for Heterodyne protocol by replacing $K^{\rm hom}_{\rm suc}(x)$ with $K^{\rm het}_{\rm suc}(\omega)$ and $V_{B;A\to A}^{\rm hom}(x)$ with $V_{B;A\to A}^{\rm het}(\omega)$ as well as using Eqs.~\eqref{eq:K_omega}, \eqref{eq:controlled_iso_het_alt}, \eqref{eq:plus_od_minus_ev}, and \eqref{eq:minus_od_plus_ev}:
\begin{align}
    \begin{split}
    M^{\rm het}_{\rm ph} &= \iint_{-\infty}^{\infty}  d\omega_r d\omega_i  \left[\hat{P}\!\left(\left[I_A\otimes\bigl(K^{\rm het}_{\rm suc}(\omega)\bigr)^{\dagger}\right] \text{C-}X_{BA}\, R^Z_{A}(2\omega_i\beta) \ket{u_+^{\rm het}(\omega_r)}_A\otimes \ket{-}_B\right) \right.\\
    & \hspace{3cm} \left. + \hat{P}\!\left(\left[I_A\otimes\bigl(K^{\rm het}_{\rm suc}(\omega)\bigr)^{\dagger}\right] \text{C-}X_{BA}\, R^Z_{A}(2\omega_i\beta)\ket{u_-^{\rm het}(\omega_r)}_A\otimes \ket{+}_B\right) \right]
    \end{split}\\
    \begin{split}
    &= \iint_{-\infty}^{\infty}\frac{2f_{\rm suc}(\omega_r)}{\pi} d\omega_r d\omega_i  \left[\hat{P}\!\left(\Pi_{AC}^{(+,\mathrm{od}), (-,\mathrm{ev})}R^Z_{A}(2\omega_i\beta)\ket{u_+^{\rm het}(\omega_r)}_A\otimes \ket{\omega}_C\right) \right.\\
    &\hspace{5cm} \left.+ \hat{P}\!\left(\Pi_{AC}^{(-,\mathrm{od}), (+,\mathrm{ev})}R^Z_{A}(2\omega_i\beta)\ket{u_-^{\rm het}(\omega_r)}_A\otimes \ket{\omega}_C\right)\right].
    \end{split}
     \label{eq:phase_error_het}
\end{align}
Using Eq.~\eqref{eq:wave_func_coherent}, we observe that
\begin{align}
    \frac{1}{\pi} \int d\omega_i \exp(\pm 2i\omega_i \beta) \ket{\omega}\!\bra{\omega} &= \frac{1}{\pi} \iiint d\omega_i dx dx' \sqrt{\frac{2}{\pi}} e^{\pm 2i\omega_i \beta -(x-\omega_r)^2 + 2i\omega_i x - (x'-\omega_r)^2 - 2i\omega_i x'} \ket{x}\!\bra{x'} \\ &= 2\iint dx dx'\, \delta(2(x\pm\beta - x')) \ket{x}\!\braket{x|\omega_r}\!\braket{\omega_r|x'}\!\bra{x'}  \\ &= \int dx \ket{x}\!\braket{x|\omega_r}\!\braket{\omega_r|x\pm \beta}\!\bra{x\pm\beta}.
\end{align}
Applying this to Eq.~\eqref{eq:phase_error_het} and changing the integration variable appropriately, we have
\begin{equation}
    \begin{split}
    M^{\rm het}_{\rm ph} &= \iint_{-\infty}^{\infty} 2f_{\rm suc}(\omega_r) d\omega_r dx \left[\hat{P}\left(\Pi_{AC}^{(+,\mathrm{od}), (-,\mathrm{ev})} O_{AC}^{\beta}(x)\ket{u_+^{\rm het}(\omega_r)}_A \otimes \ket{\omega_r}_C\right) \right.\\
    & \hspace{5cm} \left.  + \hat{P}\left(\Pi_{AC}^{(-,\mathrm{od}), (+,\mathrm{ev})} O_{AC}^{\beta}(x)\ket{u_-^{\rm het}(\omega_r)}_A \otimes \ket{\omega_r}_C\right)\right],
    \end{split} \label{eq:het_integrated}
\end{equation}
where the operator $O_{AC}^{\beta}(x)$ is defined as 
\begin{equation}
    O_{AC}^{\beta}(x)\coloneqq \ket{0}\!\bra{0}_A\otimes \ket{x}\!\bra{x}_C + \ket{1}\!\bra{1}_A\otimes \ket{x-\beta}\!\bra{x-\beta}_C.\label{eq:def_O}
\end{equation}

\subsection{Finite-size analysis} \label{sec:finite_size_analysis}
Since the phase error operator was defined on systems $A$ and $C$, we can follow essentially the same analysis as that in Ref.~\cite{matsuura2021}.  
Let us define the following operators:
\begin{align}
    \Pi^{\rm fid} &\coloneqq \ket{0}\!\bra{0}_A\otimes \ket{\beta}\!\bra{\beta}_C + \ket{1}\!\bra{1}_A\otimes \ket{-\beta}\!\bra{-\beta}_C \label{eq:def_of_pi_fid}\\
    &= \ket{\phi_-}\!\bra{\phi_-}_{AC} + \ket{\phi_+}\!\bra{\phi_+}_{AC}, \\
    \Pi^{\rm trash}_{-} &\coloneqq \ket{-}\!\bra{-}_A\otimes I_{C},
\end{align}
where $\ket{\phi_{\pm}}_{AC}$ are defined in Eqs.~\eqref{eq:phi_plus} and \eqref{eq:phi_minus}.
For any density operator $\rho$ on the joint system $AC$, these operators satisfy
\begin{align}
    \mathbb{E}_{\rho}\left[\hat{F}^{(i)}\right] &\leq p_{\rm test}\mathrm{Tr}\!\left[\rho\,\Pi^{\rm fid}\right], \label{eq:F_expectation}\\
    \mathbb{E}_{\rho}\left[\hat{Q}_{-}^{(i)}\right] &= p_{\rm trash}\mathrm{Tr}\!\left[\rho\,\Pi^{\rm trash}_{-}\right],\label{eq:Q_minus_expectation}
\end{align}
where the first inequality follows from Corollary~\ref{cor:lower_fid_coherent} in Appendix~\ref{sec:fidelity_lower_bound} as well as the definition of $\hat{F}^{(i)}$ in Estimation protocol (see Box 3), and the second equality follows from the definition of $\hat{Q}_{-}^{(i)}$ in Estimation protocol (see Box 3). 
Let $M^{\rm hom/het}[\kappa,\gamma]$ for positive numbers $\kappa$ and $\gamma$ determined prior to the protocol be defined as
\begin{equation}
    M^{\rm hom/het}[\kappa,\gamma] \coloneqq M^{\rm hom/het}_{\rm ph} + \kappa \Pi^{\rm fid} - \gamma \Pi^{\rm trash}_{-}. \label{eq:def_of_M_kappa_gamma}
\end{equation}
In Corollaries \ref{cor:ope_ineq_M} and \ref{cor:ope_ineq_M_het} in Appendix~\ref{sec:operator_ineq}, we show an inequality
\begin{equation}
    M^{\rm hom/het}[\kappa,\gamma] \leq B^{\rm hom/het}(\kappa, \gamma)\, I_{AC}
    \label{eq:ope_ineq_first}
\end{equation}
with a computable convex function $B^{\rm hom/het}(\kappa,\gamma)$.
Let $\hat{T}^{(i)}[\kappa,\gamma]$ be a linear combination of random variables at $i$th round in Estimation protocol given by
\begin{equation}
    \hat{T}^{(i)}[\kappa,\gamma] \coloneqq p_{\rm sig}^{-1}\hat{N}_{\rm ph}^{{\rm suc}\,(i)} + p_{\rm test}^{-1} \kappa \hat{F}^{(i)} - p_{\rm trash}^{-1}\gamma \hat{Q}_{-}^{(i)}.
    \label{eq:def_of_T}
\end{equation}
From Eqs.~\eqref{eq:N_suc_expectation}, \eqref{eq:F_expectation}, \eqref{eq:Q_minus_expectation}, and \eqref{eq:ope_ineq_first}, we have
\begin{equation}
    \mathbb{E}_{\rho}\!\left[\hat{T}^{(i)}[\kappa,\gamma]\right] \leq \mathrm{Tr}\left[\rho M^{\rm hom/het}[\kappa,\gamma]\right] \leq B^{\rm hom/het}(\kappa, \gamma),
\end{equation}
for any density operator $\rho$.
Furthermore, from the definition of $\hat{N}^{{\rm suc}\,(i)}_{\rm ph}$, $\hat{F}^{(i)}$, and $\hat{Q}_-^{(i)}$ in Estimation protocol (see Box 3), we have
\begin{equation}
    \min\Bigl\{p_{\rm test}^{-1}\kappa \,\min_{\nu\geq 0}\Lambda_{m,r}(\nu),-p_{\rm trash}^{-1}\gamma\Bigr\} \leq \hat{T}^{(i)}[\kappa,\gamma] \leq \max\Bigl\{p_{\rm sig}^{-1}, p_{\rm test}^{-1}\kappa \,\max_{\nu\geq 0}\Lambda_{m,r}(\nu)\Bigr\}.
\end{equation}
Let $\delta_1(\epsilon)$ be defined as
\begin{equation}
    \delta_1(\epsilon) \coloneqq \left(\max\Bigl\{p_{\rm sig}^{-1}, p_{\rm test}^{-1}\kappa \,\max_{\nu\geq 0}\Lambda_{m,r}(\nu)\Bigr\} - \min\Bigl\{p_{\rm test}^{-1}\kappa \,\min_{\nu\geq 0}\Lambda_{m,r}(\nu),-p_{\rm trash}^{-1}\gamma\Bigr\}\right)\sqrt{\frac{N}{2}\ln\!\(\frac{1}{\epsilon}\)}. \label{eq:def_delta_1}
\end{equation}
Then, by applying Proposition~\ref{prop:azuma_doob} in Appendix~\ref{sec:azuma_doob} to $\{\hat{T}^{(k)}[\kappa,\gamma]\}_{k=1,\ldots,N}$ as well as using Eqs.~\eqref{eq:def_of_T}--\eqref{eq:def_delta_1}, we have
\begin{equation}
    \mathrm{Pr}\left[\left(\sum_{k=1}^{N}\hat{T}^{(k)}[\kappa,\gamma] =\right) p_{\rm sig}^{-1}\hat{N}_{\rm ph}^{\rm suc} + p_{\rm test}^{-1} \kappa \hat{F} - p_{\rm trash}^{-1}\gamma \hat{Q}_{-} \geq N B^{\rm hom/het}(\kappa,\gamma) + \delta_1(\epsilon/2)\right] \leq \epsilon/2.
    \label{eq:ineq_for_T}
\end{equation}
Since $\hat{Q}_{-}$ is determined solely by Alice's qubits, each in the state $\mathrm{Tr}_{\tilde{C}}(\ket{\Phi}\!\bra{\Phi}_{A\tilde{C}})$ with $\ket{\Phi}_{A\tilde{C}}$ given in Eq.~\eqref{eq:prepared_state}, it follows the same statistics as a tally of $\hat{N}^{\rm trash}$ Bernoulli trials with a probability $q_{-}\coloneqq \|\bra{-}_A\ket{\Psi}_{A\tilde{C}}\|^2 = (1-e^{-2\mu})/2$.  Hence, by using the Chernoff-Hoeffding inequality~\cite{Hoeffding1963,Cover2012}, we have
\begin{equation}
    \mathrm{Pr}\left[\hat{Q}_{-} \geq q_- \hat{N}^{\rm trash} + \delta_{2}(\epsilon/2;\hat{N}^{\rm trash})\right] \leq \epsilon/2.
    \label{eq:ineq_for_Q}
\end{equation}
Here, $\delta_2(\epsilon;n)$ is defined as \cite{matsuura2021}
\begin{equation}
    \begin{cases}
        D(q_{-}+\delta_{2}(\epsilon;n)/n \| q_-) = -\frac{1}{n}\log_2(\epsilon) & (\epsilon > q_-^n) \\
        \delta_{2}(\epsilon;n) = (1 - q_-)n & (\epsilon \leq q_-^n)
    \end{cases},
\end{equation}
where 
\begin{equation}
    D(x\|y) \coloneqq x\log_2\frac{x}{y} + (1 - x)\log_2\frac{1 - x}{1 - y}
    \label{eq:kl_divergence}
\end{equation}
is the Kullback-Leibler divergence.
Combining Eqs.~\eqref{eq:ineq_for_T}, and \eqref{eq:ineq_for_Q}, by setting 
\begin{equation}
    U(\hat{F},\hat{N}^{\rm trash}) = p_{\rm sig}\bigl(N B^{\rm hom/het}(\kappa,\gamma) + \delta_1(\epsilon/2)\bigr) - \frac{p_{\rm sig}}{p_{\rm test}} \kappa \hat{F} + \frac{p_{\rm sig}}{p_{\rm trash}}\gamma \bigl(q_- \hat{N}^{\rm trash} + \delta_{2}(\epsilon/2;\hat{N}^{\rm trash})\bigr),
\end{equation}
we observe that Eq.~\eqref{eq:probability_condition} holds from the union bound.

\section{Numerical simulations} \label{sec:numerical_simulations}
We compute (the lower bound on) the net key gain per pulse (i.e., key rate $\hat{G}$) against the transmission distance with various values of excess noise at the channel output.  In this model, Bob receives Gaussian states $\rho_{\rm model}^{(\hat{a})}$ obtained by randomly displacing attenuated coherent states $\ket{(-1)^{\hat{a}} \sqrt{\eta\mu}}$ with attenuation rate $\eta$ to increase their variances via factor of $(1+\xi)$, i.e.,
\begin{equation}
    \rho_{\rm model}^{(\hat{a})} \coloneqq \frac{2}{\pi \xi}\int_{\mathbb{C}} e^{-2|\gamma|^2/\xi}\ket{(-1)^{\hat{a}} \sqrt{\eta\mu} + \gamma}\!\bra{(-1)^{\hat{a}} \sqrt{\eta\mu} + \gamma} d^2\gamma.
    \label{eq:channel_output}
\end{equation}
For simplicity, the number $N_{\rm smp}$ of the sampling rounds is set to be $N/100$, and the bit error correction efficiency $f$ in Eq.~\eqref{eq:bit_error_fraction} is to be $0.95$ \footnote{Currently, this level of efficiency may be too optimistic because the bit error correction in our protocol must succeed with probability no smaller than $1-\varepsilon_{\rm cor}/2$ without the use of the verification.}.  The acceptance probability $f_{\rm suc}(x)$ is assumed to be a step function $\Theta(x-x_{\rm th})$ with a threshold $x_{\rm th}(>0)$, where $\Theta(x)$ denotes the Heaviside step function.  The expected amplitude of the coherent state $\beta$ is chosen to be $\sqrt{\eta\mu}$.  We set the security parameter $\varepsilon_{\rm sec}=2^{-50}$, and set $\varepsilon_{\rm cor}=\varepsilon_{\rm sec}/2$ and $\epsilon = 2^{-s} = \varepsilon_{\rm sec}^2/16$.  

We assume that the number of ``success'' signal rounds $\hat{N}^{\rm suc}$ is equal to its expectation, i.e.,
\begin{align}
    \mathbb{E}[\hat{N}^{\rm suc}]  &= p_{\rm sig}N\int_{-\infty}^{\infty} (f_{\rm suc}(x) + f_{\rm suc}(-x))\bra{x}\frac{1}{2}\sum_{a\in\{0,1\}}\rho_{\rm model}^{(a)}\ket{x} dx \\
    \begin{split}
    &= p_{\rm sig}N\int_{-\infty}^{\infty} \frac{1}{2}\left[(f_{\rm suc}(x) + f_{\rm suc}(-x))\bra{x}\rho_{\rm model}^{(0)}\ket{x} \right. \\
    &\hspace{3.5cm}\left. + (f_{\rm suc}(-x) + f_{\rm suc}(x))\bra{-x}\rho_{\rm model}^{(1)}\ket{-x}\right] dx  \end{split}\\
    &= p_{\rm sig}N\int_{-\infty}^{\infty} \frac{1}{2}\sum_{a\in\{0,1\}}(f_{\rm suc}(x) + f_{\rm suc}(-x))\bra{(-1)^a x}\rho_{\rm model}^{(a)}\ket{(-1)^a x}  dx \\
    &= p_{\rm sig}N(P^+_{\rm hom} + P^-_{\rm hom}),
\end{align}
where 
\begin{align}
    P^{\pm}_{\rm hom} &\coloneqq \int_{-\infty}^{\infty}  f_{\rm suc}(\pm x)\bra{(-1)^a x} \rho_{\rm model}^{(a)} \ket{(-1)^a x} dx \\
    &= \int_{x_{\rm th}}^{\infty}  \bra{\pm(-1)^a x} \rho_{\rm model}^{(a)} \ket{\pm(-1)^a x} dx \\
    &= \frac{1}{2}{\rm erfc}\biggl((x_{\rm th} \mp \sqrt{\eta\mu})\sqrt{\frac{2}{1+\xi}} \biggr),
\end{align}
for Homodyne protocol \cite{matsuura2021}.  In the above, $P^{+}_{\rm hom}$ (resp.~$P^{-}_{\rm hom}$) denotes the probability that Alice and Bob obtain coincident (resp.~incoincident) bit values in the signal round.
(Note that the apparent $a$-dependence in the definition of $P^{\pm}_{\rm hom}$ above is not the case because of the symmetry of $\rho_{\rm model}^{(a)} $ defined in Eq.~\eqref{eq:channel_output}.)
In a similar way, we obtain for Heterodyne protocol that 
\begin{align}
    \mathbb{E}[\hat{N}^{\rm suc}] &= p_{\rm sig}N(P^+_{\rm het} + P^-_{\rm het}), \\
    P^{\pm}_{\rm het} &\coloneqq \iint_{-\infty}^{\infty}  \frac{f_{\rm suc}(\pm \omega_r)}{\pi}\bra{(-1)^a \omega} \rho_{\rm model}^{(a)} \ket{(-1)^a \omega} d\omega_r d\omega_i \\
    &= \int_{-\infty}^{\infty}d\omega_i \int_{x_{\rm th}}^{\infty} d\omega_r \frac{f_{\rm suc}(\pm \omega_r)}{\pi}\bra{(-1)^a \omega} \rho_{\rm model}^{(a)} \ket{(-1)^a \omega} \\
    &=\frac{1}{2}{\rm erfc}\biggl((x_{\rm th} \mp \sqrt{\eta\mu})\sqrt{\frac{2}{2+\xi}} \biggr).
\end{align}
We also assume that the number of ``success'' sampling rounds is equal to $(P^+_{\rm hom/het} + P^-_{\rm hom/het})N_{\rm smp}$, the number of test rounds $\hat{N}^{\rm test}$ is equal to $p_{\rm test}N$, and the number of trash rounds $\hat{N}^{\rm trash}$ is equal to $p_{\rm trash}N$.  The test outcome $\hat{F}$ is assumed to be equal to its expectation given by \cite{matsuura2021}
\begin{align}
    \mathbb{E}[\hat{F}] &= p_{\rm test}N\, \frac{1}{2}\sum_{a\in\{0,1\}}\mathbb{E}_{\rho_{\rm model}^{(a)}}[\Lambda_{m,r}(|\hat{\omega} - (-1)^a\sqrt{\eta\mu}|^2)] \\
    &= \frac{p_{\rm test}N}{1 + \xi/2}\left[1 - (-1)^{m+1}\left(\frac{\xi/2}{1 + r(1 + \xi/2)}\right)^{m+1}\right].
\end{align}
For the test function $\Lambda_{m,r}$ in the above, we adopt $m=1$ and $r=0.4120$, which leads to $(\max_{\nu\geq 0} \Lambda_{m,r}(\nu), \min_{\nu\geq 0} \Lambda_{m,r}(\nu))=(2.824,-0.9932)$.
We assume that the number $\hat{E}_{\rm obs}$ of bit errors observed in the ``success'' sampling rounds is equal to its expectation $\hat{E}_{\rm obs}=P^-_{\rm hom/het} N_{\rm smp}$.   
The upper-bound $e_{\rm qber}$ on the bit error rate is thus given by Eq.~\eqref{eq:upper_bit_error_rate} with the parameters $\hat{N}^{\rm suc}$, $\hat{N}^{\rm suc}_{\rm smp}$, and $\hat{E}_{\rm obs}$ given above.
Under these assumptions, the remaining parameters to be determined are six parameters $(\mu, x_{\rm th}, p_{\rm sig}, p_{\rm test}, \kappa, \gamma)$.  We determined $(\kappa,\gamma)$ via a convex optimization using CVXPY 1.2.1 and $(\mu,x_{\rm th},p_{\rm sig}, p_{\rm test})$ via the Nelder-Mead in the scipy.minimize library in Python, for each transmission distance $L$ with the attenuation rate $\eta$ assumed to be $10^{-0.02L}$.

\begin{figure*}[t]
    \centering
    \includegraphics[width=0.99\linewidth]{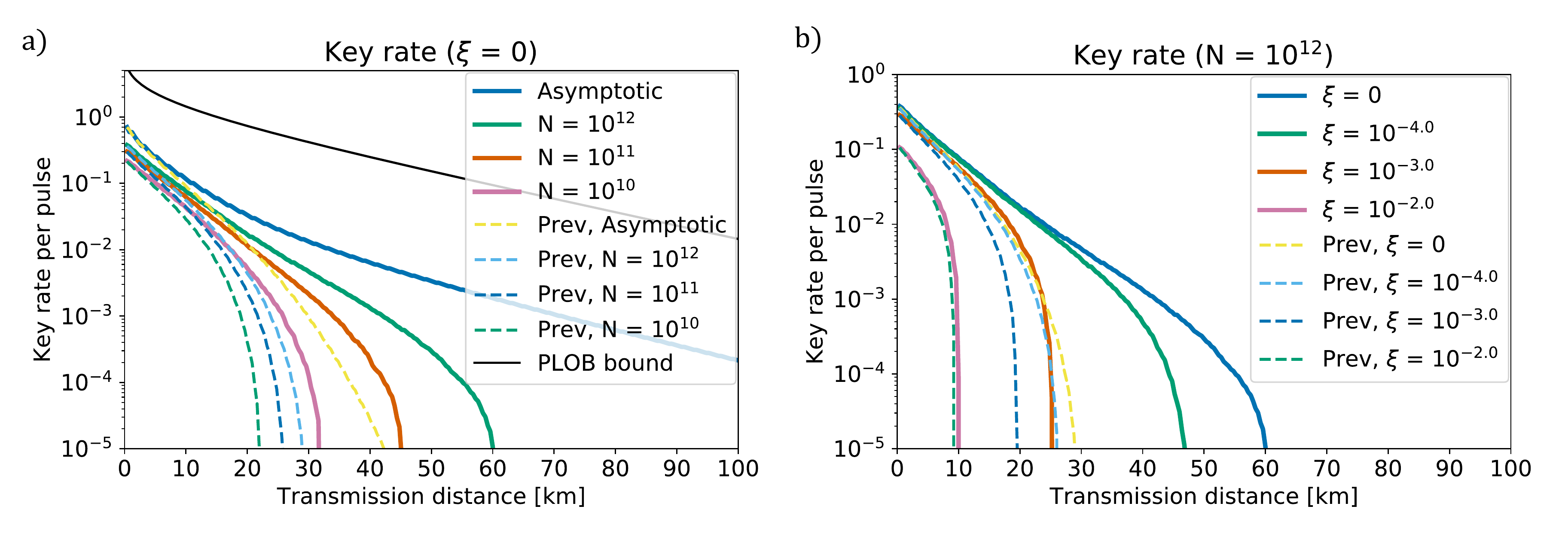}
    \caption{Key rates of the Homodyne protocol against transmission distance over an optical fiber.  The attenuation rate of the optical fiber is assumed to be $10^{-0.02L}$ with transmission distance $L$ km, an error correction efficiency $f$ in Eq.~\eqref{eq:bit_error_fraction} is set to be $0.95$, and the number of sampling rounds $N_{\rm smp}$ is set to be $N/100$.  a) Key rates when the excess noise $\xi$ at the channel output is zero; that is, the channel is pure loss.  The bold solid lines show the key rates with our refined analysis developed here, the broken lines show those with the previous analysis \cite{matsuura2021}, and the black thin line shows the PLOB bound, which is the ultimate limit of the key rate of one-way QKD \cite{pirandola2017}.  One can see that the logarithm of the asymptotic key rate decreases in parallel to the PLOB bound with our refined analysis against the transmission distance ($\gg 1$ km) as opposed to the previous results \cite{matsuura2021}.  Improvement in the key rate is sustained in the finite-size case.   b) Key rates when $N=10^{12}$ with various values of excess noise parameter $\xi$.  (The detail of the noise model is given in the main text.)  The solid lines show the key rates with our refined analysis, and the broken lines show those with the previous results \cite{matsuura2021}.  One can see that, although the key rate significantly improves for the pure-loss channel, the excess noise as high as $\xi=10^{-3}$--$10^{-2}$ degrades the performance to almost the same level as that of the previous results.}
    \label{fig:key_rates_hom}
\end{figure*}

Figure~\ref{fig:key_rates_hom} shows the key rates of Homodyne protocol for the channel model explained above.  Figures show that under the condition of low excess noise, our refined analysis results in significantly higher key rates and longer transmission distance than that of the previous results~\cite{matsuura2021} even in the finite-key case.  Furthermore, the logarithm of the asymptotic key rate in the pure-loss case (i.e., $\xi=0$) is in parallel to the PLOB bound \cite{pirandola2017} against the transmission distance; that is, it achieves a linear scaling against the channel transmission, which is known to be optimal for one-way QKD in the pure-loss channel.  When the excess noise $\xi$ is around $10^{-3.0}$--$10^{-2.0}$, however, the improvements in our refined analysis are lost.  The result of the parameter optimization implies that our refined analysis generates the key with relatively small intensity $\mu$ of the input coherent states compared to the previous analyses; e.g., the optimized input intensity $\mu$ of Homodyne protocol is $\sim 0.04$ in our refined analysis compared to $\sim 0.2$ in the previous analysis \cite{matsuura2021} at $\eta=0.1$ (i.e., 50 km) for the asymptotic pure-loss case.

The key rate of Heterodyne protocol has a similar behavior.  Figure~\ref{fig:key_rates_het} shows the key rates of Heterodyne protocol with the same noise model as above.  Figures show that our refined analysis significantly improves the key rate against the pure-loss channel, but is fragile against excess noise.  One can see, however, that, while the key rate of Heterodyne protocol is still low compared to that of Homodyne protocol, the achievable distance (i.e., the distance with a non-zero key rate) now becomes comparable with our refined analysis.  This implies that our refined analysis based on the reverse reconciliation is more effective for Heterodyne protocol.

\begin{figure*}[t]
    \centering
    \includegraphics[width=0.99\linewidth]{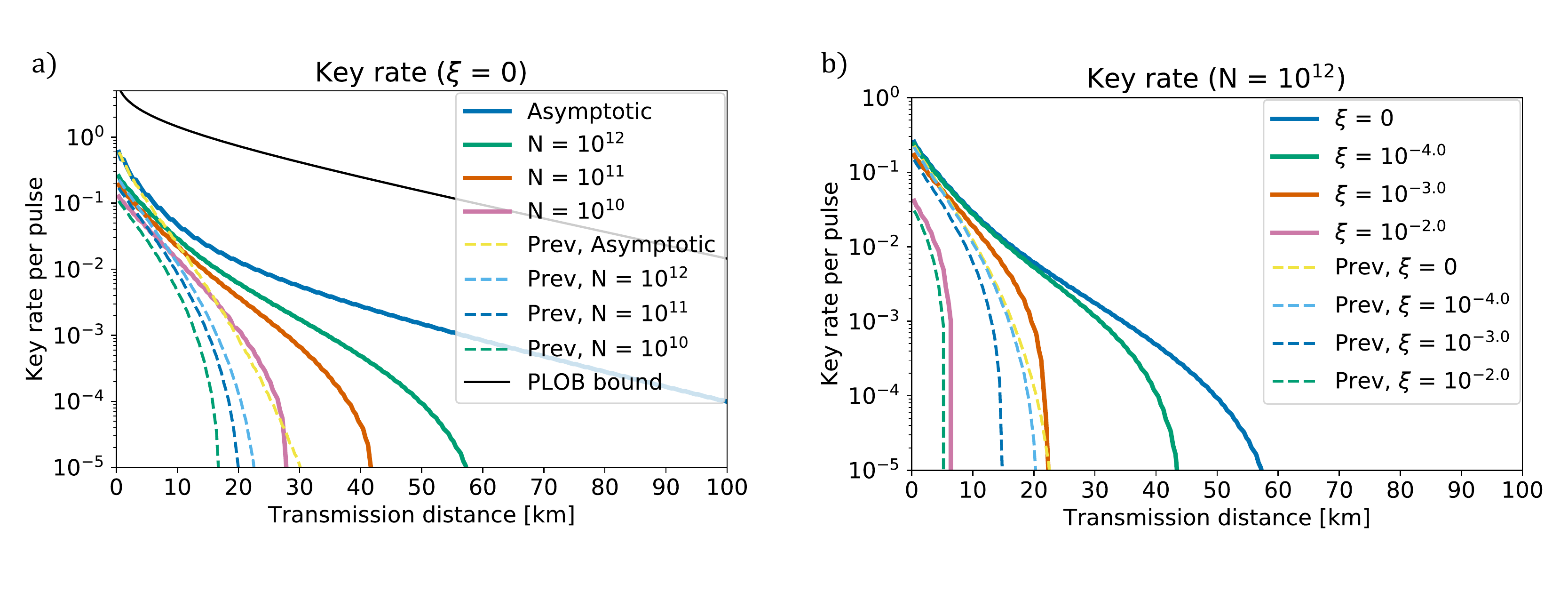}
    \caption{Key rates of the Heterodyne protocol against transmission distance over an optical fiber.  The noise models are the same as those of Homodyne protocol.  a) Key rates when the excess noise $\xi$ at the channel output is zero; that is, the channel is pure loss.  The bold solid lines show the key rates with our refined analysis developed here, the broken lines show those with the previous analysis \cite{Yamano2022}, and the black thin line shows the PLOB bound, which is the ultimate limit of the key rate of one-way QKD \cite{pirandola2017}.  One can see that the logarithm of the asymptotic key rate is in parallel to the PLOB bound when the transmission distance is large in the same way as that of Homodyne protocol.  The key rate is still less (about half) than that of Homodyne protocol.   b) Key rates when $N=10^{12}$ with various values of excess noise parameter $\xi$.  The solid lines show the key rates with our refined analysis, and the broken lines show those with the previous result \cite{Yamano2022}. 
    }
    \label{fig:key_rates_het}
\end{figure*}

\section{Discussion} \label{sec:discussion}
We propose a refined security analysis for the protocol proposed in Ref.~\cite{matsuura2021} based on the reverse reconciliation.  The motivating ideas of our refinement come from the facts that the distillability of a secret key from a quantum state is a looser condition than the distillability of an entanglement from it \cite{horodecki2006, Renes2007, horodecki2008, renner2008,horodecki2009, complementarity} and that the reverse reconciliation can increase the key rate for CV QKD protocols \cite{Silberhorn2002}.
To exploit the ideas, we developed the procedure of ``twisting'' Alice's system with $V^{\rm hom}_{B; A\to R}(x)$ (resp.~$V^{\rm het}_{B; A\to R}(\omega)$) controlled by Bob's qubit.  Similar techniques have already appeared in previous works \cite{horodecki2006, renner2008, horodecki2008, horodecki2009, complementarity, Bourassa2020}.  
Our finding is that by using the twisting operation that minimizes the phase error probability for the pure-loss channel, the protocol has asymptotically optimal scaling in the key rates both for Homodyne and Heterodyne protocols.  
This is a clear distinction from the previous results \cite{matsuura2021, Yamano2022}; there, the asymptotic key rate non-linearly decreases against the channel transmission. 
This results challenge conventionally considered limitations on the phase error correction applied to the CV QKD.  For example, in Ref.~\cite{Furrer2014reverse}, the security proof based on the entropic uncertainty relation and the phase error correction did not achieve optimal linear key-rate scaling in the asymptotic case even with reverse reconciliation.  The author ascribed this bad performance to the gap between a conditional entropy with a quantum conditioning and one with a classical conditioning, where the latter corresponds to measuring the quantum system with a fixed basis.  Here, our security proof can be seen from Eqs.~\eqref{eq:controlled_iso_hom_alt}, \eqref{eq:controlled_iso_het_alt} and the definitions of $\ket{u_\pm^{\rm hom/het}(x)}_A$ as ``twisting'' the (measurement) basis in which Alice's system is reduced to classical depending on the expected amplitude of a received optical signal as well as an observed homodyne/heterodyne outcome. This may minimize the above mentioned gap and thus achieve the asymptotically optimal scaling.  In this sense, the conventionally-thought limitation may not be an actual limitation.

The improvement in the performance remains in the finite-key case but is lost under the existence of excess noise as high as $\xi=10^{-3}\text{--}10^{-2}$ at the channel output.  
This may limit the feasibility of our binary-modulation protocol, but current theoretical progress in CV QKD reveals that the discrete-modulation CV-QKD protocols with four types of modulation have more tolerance against excess noise than those with binary modulation \cite{four_state_Leverrier,four_state_Lutken,liu2021}.  
What is important is that our security proof can be extended to the four-state protocols with binary outcomes, such as Protocol 2 in Ref.~\cite{four_state_Lutken} and a protocol in Ref.~\cite{liu2021}, by replacing the bit-extracting measurements of these protocols with the qubit-extracting maps as shown in Eq.~\eqref{eq:I_homo} and constructing the corresponding phase error operator.   This is, however, much more complicated than the previous analysis, and we leave the problem as future work.

There are several remaining questions with our present results.  
The first and foremost is whether we can obtain higher tolerance against excess noise by extending our analysis to the four-state protocols.  
As explained above, our analysis can be extended to the four-state protocols with binary outputs \cite{four_state_Lutken,liu2021}, i.e., protocols that use homodyne measurement to distinguish signals.  
With the same type of argument based on the phase error estimation, we can carry out the finite-size security proof for these protocols in principle.  
However, developing the analyses that preserve the robustness against excess noise for these protocols still has non-triviality.  
A more challenging problem is to apply our finite-size security proof to the four-state protocols with more than two outputs, such as a protocol in Ref.~\cite{four_state_Leverrier} and Protocol 1 in Ref.~\cite{four_state_Lutken}.  
In this case, the definition of phase errors is already non-trivial as opposed to those with binary outputs, and we have to develop more elaborated finite-size security proof.  
Whether we can extend our techniques to these protocols or protocols with even more constellations \cite{denys2021} is still open.

Another important theoretical question is whether the trusted-noise model can be applied to our security analysis.  In practice, even the excess noise of $\xi=10^{-3}$ at the channel output is difficult to realize if all the noises are untrusted.  Recently, efforts have been made in the field of CV QKD on how to incorporate noises that are intrinsic to apparatuses and thus inaccessible to Eve into the security proof as trusted noises.  This effectively eases the requirement on the experimental apparatuses.  In the present security analysis as well as ones in Refs.~\cite{matsuura2021, Yamano2022}, the fidelity test measures the fidelity to a pure coherent state, which cannot be naively generalized to the fidelity to a mixed state.  Whether we can incorporate trusted noises into the fidelity test may be crucial in this direction.

From the viewpoint of the feasibility of the protocol, the total number of $10^{12}$ of rounds to obtain a tolerable finite-size performance may be demanding.  
The finite-size performance may be improved by applying recently developed refinement \cite{kato2020} of the Azuma's inequality \cite{Azuma_ineq} that utilizes unconfirmed knowledge.  
What is non-trivial for the application of this is that the random variable in our application of Azuma's inequality can not directly be observed even at the end of the protocol.  
Whether we can apply the refined concentration inequality \cite{kato2020} with the information accessible in our protocol (in a similar fashion to Ref.~\cite{zhou2021}) may be an interesting problem.

\begin{acknowledgments}
    This work was supported by the Ministry of Internal Affairs and Communications (MIC), R\&D of ICT Priority Technology Project (grant number JPMI00316); Cross-ministerial Strategic Innovation Promotion Program (SIP) (Council for Science, Technology and Innovation (CSTI)); JSPS KAKENHI Grant Number JP22K13977; JSPS Overseas Research Fellowships.
\end{acknowledgments}

\bigskip
\bigskip

\subsection*{Code availability}
Computer codes to calculate the key rates are available from the corresponding author upon reasonable request.

\appendix

\section{Definition of (composable) security} \label{sec:definition_security}
In this section, we give a definition of security, which can be divided into two conditions, secrecy and correctness.  The definition satisfies the so-called universal composability \cite{composability}.
Since the correctness is satisfied by the procedure in the actual protocol, it is the secrecy that is the focus of our security proof.
In this section, when the index to denote a system is dropped from a joint quantum state, the partial trace is implicitly taken for that system; e.g., for the state $\rho_{AB}$, $\rho_A=\mathrm{Tr}_B[\rho_{AB}]$.

Let us assume that a given (prepare-and-measure) QKD protocol generates the final key bits $\hat{\bm{z}}_A^{\rm fin}$ and $\hat{\bm{z}}_B^{\rm fin}$ between Alice and Bob with the length denoted as $\hat{N}^{\rm fin}$.  When the protocol generates null key or is aborted, we set $\hat{\bm{z}}_A^{\rm fin}=\hat{\bm{z}}_B^{\rm fin}=\emptyset$ and $\hat{N}^{\rm fin}=0$.  Given the final key length $\hat{N}^{\rm fin}=N\geq 1$, let $\rho_{ABE|N}^{\rm fin}$ be a quantum state between Alice's and Bob's $N$-bit final key as well as Eve's system at the end of the actual protocol.  
\begin{definition}
    Given an arbitrary positive real number $\varepsilon_{\rm sec}>0$, a QKD protocol is called $ \varepsilon_{\rm sec}$-secure if it satisfies the following:
    \begin{equation}
        \frac{1}{2}\sum_{N\geq 1}\mathrm{Pr}(\hat{N}^{\rm fin}=N) \left\|\rho_{ABE|N}^{\rm fin} - \rho_{ABE|N}^{\rm ideal}\right\|_1 \leq \varepsilon_{\rm sec},
    \end{equation}
    where $\rho_{ABE|N}^{\rm ideal}$ is defined as
    \begin{equation}
        \rho_{ABE|N}^{\rm ideal} \coloneqq \sum_{\bm{z}\in\{0,1\}^N} \frac{1}{2^N}\ket{\bm{z}}\!\bra{\bm{z}}_A\otimes \ket{\bm{z}}\!\bra{\bm{z}}_B \otimes \rho_{E|N}^{\rm fin}.
    \end{equation}
\end{definition}

\begin{prop}\label{prop:correctness_secrecy}
    For arbitrary positive real numbers $\varepsilon_{\rm cor}>0$ and $\varepsilon_{\rm sc}>0$, let us assume that the following two conditions are satisfied:
    \begin{align}
        \mathrm{Pr}(\hat{N}^{\rm fin}\geq 1, \hat{\bm{z}}_A^{\rm fin}\neq \hat{\bm{z}}_B^{\rm fin}) &\leq \varepsilon_{\rm cor}, \label{eq:correctness_cond}\\
        \frac{1}{2}\sum_{N\geq 1}\mathrm{Pr}(\hat{N}^{\rm fin}=N) \left\|\rho_{AE|N}^{\rm fin} - \rho_{AE|N}^{\rm ideal}\right\|_1 &\leq \varepsilon_{\rm sc}.\label{eq:secrecy_cond}
    \end{align}
    Then, the protocol is $(\varepsilon_{\rm cor}+\varepsilon_{\rm sc})$-secure as long as the protocol uses direct reconciliation.
\end{prop}
\begin{proof}
    See e.g.~Ref.~\cite{complementarity}.
\end{proof}
The condition Eq.~\eqref{eq:correctness_cond} is called the correctness condition and the condition Eq.~\eqref{eq:secrecy_cond} is called the secrecy condition.
\begin{remark}
    Proposition~\ref{prop:correctness_secrecy} holds for protocols with reverse reconciliation if the system $A$ is replaced with $B$ in the secrecy condition Eq.~\eqref{eq:secrecy_cond}.
\end{remark}

\section{Lemmas and Theorems used in the main text}

\subsection{Lower bound on the fidelity to a coherent state} \label{sec:fidelity_lower_bound}
This section summarizes the technical result of Ref.~\cite{matsuura2021}.
\begin{theorem} \label{thm:laguerre}
    For an integer $m\geq 0$ and a real number $r>0$, let $\Lambda_{m,r}(\nu)$ $(\nu\geq 0)$ be a bounded function defined as
    \begin{equation}
        \label{eq:postprocess_function}
        \Lambda_{m,r}(\nu)\coloneqq e^{-r \nu}(1+r)L_m^{(1)}((1+r)\nu),
    \end{equation}
    where $L_m^{(k)}$ is the associated Laguerre polynomials given by
    \begin{align}
      \label{eq:associate_Laguerre}
        L_n^{(k)}(\nu)&\coloneqq (-1)^k\frac{d^kL_{n+k}(\nu)}{d\nu^k},\\
        L_n(\nu)&\coloneqq \frac{e^{\nu}}{n!}\frac{d^n}{d\nu^n}(e^{-\nu}\nu^n).
    \end{align}
    Then, we have
    \begin{equation}
    \label{eq:theorem1_estimation}
    \mathbb{E}_\rho[\Lambda_{m,r}(|\hat{\omega}|^2)]= \iint_{\omega\in\mathbb{C}}\frac{d^2\omega}{\pi}\Lambda_{m,r}(|\omega|^2)\bra{\omega}\rho\ket{\omega} =\bra{0}\rho\ket{0}+\sum_{n=m+1}^{\infty}\frac{\bra{n}\rho\ket{n}}{(1+r)^n}I_{n,m},
\end{equation}
where $I_{n,m}$ are constants satisfying $(-1)^m I_{n,m}>0$.
\end{theorem}
As a Corollary, we obtain the following.
\begin{cor}\label{cor:lower_fid_coherent}
        Let $\ket{\beta}$ $(\beta\in \mathbb{C})$ be the coherent state with amplitude $\beta$.  Then, for any $\beta \in \mathbb{C}$ and for any odd positive integer $m$, we have
        \begin{equation}
          \mathbb{E}_{\rho}[\Lambda_{m,r}(|\hat{\omega}-\beta|^2)]\leq \bra{\beta}\rho\ket{\beta}.
        \end{equation}
\end{cor}
The proofs of Theorem~\ref{thm:laguerre} and Corollary~\ref{cor:lower_fid_coherent} are given in Ref.~\cite{matsuura2021}.

\subsection{Bit error sampling} \label{sec:estimate_num_of_err}
In this section, we summarize how to determine an upper bound on the bit error rate from the given sample.  As explained in the main text, $N_{\rm smp}$ sampling rounds are randomly inserted in the actual protocol in which Alice and Bob announce their bit values if Bob's detection succeeds (in the same way as in the signal round).  The number of ``success'' sampling rounds is denoted by $\hat{N}_{\rm smp}^{\rm suc}$, and the observed number of discrepancies between Alice and Bob is denoted by $\hat{E}_{\rm obs}$.

Let us first introduce a Chernoff-type bound for the hypergeometric distribution.
\begin{lemma}[Tail bound for the hypergeometric distribution \cite{chvatal1979}]\label{lem:hypergeometric} 
    Let $X_1,\ldots ,X_N$ be a binary sequence, and $M$ be the number of elements with $X_i=1$, i.e, $M\coloneqq\sum_{i=1}^N X_i$.  Let $\hat{Y}_1,\ldots,\hat{Y}_n$ $(n\leq N)$ be randomly sampled from $X_1,\ldots,X_N$ without replacement.  Let $\hat{m}\coloneqq \sum_{i=1}^n \hat{Y}_i$ be the number of ones in $\hat{Y}_1,\ldots,\hat{Y}_n$.  Then, for any $\delta \in[0, M/N]$, the following inequality holds:
    \begin{equation}
        \mathrm{Pr}\left(\frac{\hat{m}}{n} \leq \frac{M}{N} - \delta\right) \leq 2^{-n D\left(\frac{M}{N}-\delta \middle\|\frac{M}{N}\right)},
    \end{equation}
    where $D(\cdot\|\cdot)$ is defined in Eq.~\eqref{eq:kl_divergence}.
\end{lemma}
Then, the following corollary is essential for the bit error sampling.
\begin{cor}[Estimation by the simple random sampling without replacement] \label{cor:random_sampling}
    Let $X_1,\ldots ,X_N$ be a binary sequence with $M\coloneqq\sum_{i=1}^N X_i$.  Let $\hat{Y}_1,\ldots,\hat{Y}_n$ $(n\leq N)$ be randomly sampled from $X_1,\ldots,X_N$ without replacement, and define $\hat{m}\coloneqq \sum_{i=1}^n \hat{Y}_i$.  Then, for any $\epsilon\in(0,1)$, the following inequality holds:
    \begin{equation}
        \mathrm{Pr}\left(\tilde{M}_{N,n,\epsilon}(\hat{m})  < M\right) \leq \epsilon,
        \label{eq:sampling_upper_violation}
    \end{equation}
    where the function $\tilde{M}_{N,n,\epsilon}(m)$ is defined to satisfy
    \begin{equation}
        \frac{m}{n} \leq \frac{\tilde{M}_{N,n,\epsilon}(m)}{N} \leq 1 
        \label{eq:trivial_inequality}
    \end{equation}
    and for $0\leq m <n$,
    \begin{equation}
        D\left(m/n \middle\|\tilde{M}_{N,n,\epsilon}(m) / N\right) = -\frac{1}{n}\log\epsilon.
        \label{eq:upper_bound_sampling}
    \end{equation}
\end{cor}

\begin{proof}
    Let $f(M)$ be a function of $M$ satisfying $0\leq f(M)/n\leq M/N$.
    Then, from Lemma~\ref{lem:hypergeometric}, we have
    \begin{equation}
        \mathrm{Pr}\left(\frac{\hat{m}}{n} \leq \frac{M}{N} - \left(\frac{M}{N} - \frac{f(M)}{n}\right)\right) \leq 2^{-n D\left(\frac{f(M)}{n} \middle\| \frac{M}{N} \right)}.
        \label{eq:apply_hypergeometric}
    \end{equation}
    We set the function $f(M)$ to the restriction of the function $f_{N,n,\epsilon}(\bar{M})$ of the real number $\bar{M}$ that satisfies
    \begin{equation}
        D\left(f_{N,n,\epsilon}(\bar{M})/n \| \bar{M}/N \right)=-\frac{1}{n}\log \epsilon,
        \label{eq:condition_of_f}
    \end{equation}
    for $\bar{M} \in [(1-\sqrt[n]{\epsilon})N, N)$.
    The function $f_{N,n,\epsilon}(\bar{M})$ increases monotonically with increasing $\bar{M}$ in $[(1-\sqrt[n]{\epsilon})N, N)$, and its image lies in $[0,n)$. 
    Thus, from Eq.~\eqref{eq:apply_hypergeometric}, we have
    \begin{equation}
        \mathrm{Pr}\left(f^{-1}_{N,n,\epsilon}(\hat{m})\leq M\right) \leq \epsilon
        \label{eq:stronger_condition}
    \end{equation}
    for any $\hat{m}\in [0,n)$.  We define the function $\tilde{M}_{N,n,\epsilon}(m)\coloneqq f^{-1}_{N,n,\epsilon}(m)$ for $m\in [0,n)$.
    To incorporate the case $\hat{m}=n$, we use the following weaker condition that trivially follows from Eq.~\eqref{eq:stronger_condition}:
    \begin{equation}
        \mathrm{Pr}\left(\tilde{M}_{N,n,\epsilon}(\hat{m}) < M\right) \leq \epsilon,
    \end{equation}
    and define $\tilde{M}_{N,n,\epsilon}(n)=N$ so that the above holds also for $\hat{m}=n$.
    These show that Eq.~\eqref{eq:sampling_upper_violation} holds while $\tilde{M}_{N,n,\epsilon}(m)$ satisfies Eqs.~\eqref{eq:trivial_inequality} and \eqref{eq:upper_bound_sampling} by construction in Eq.~\eqref{eq:condition_of_f}.
\end{proof}

With Corollary~\ref{cor:random_sampling}, we can bound the number of total bit-error events from the sample under the given failure probability $\varepsilon_{\rm cor}/2$ by setting $N=\hat{N}^{\rm suc}+\hat{N}^{\rm suc}_{\rm smp}$, $n= \hat{N}^{\rm suc}_{\rm smp}$, and $\epsilon=\varepsilon_{\rm cor}/2$ for $\tilde{M}_{N,n,\epsilon}$.  As a result, we have the following statement; the number $E$ of bit errors in $\hat{N}^{\rm suc}$-bit sifted key is bounded from above by
\begin{equation}
    \mathrm{Pr}\left(E \leq \tilde{M}_{\hat{N}^{\rm suc}+\hat{N}^{\rm suc}_{\rm smp},\hat{N}^{\rm suc}_{\rm smp},\varepsilon_{\rm cor}/2}(\hat{E}_{\rm obs}) - \hat{E}_{\rm obs} \right) \geq 1 - \varepsilon_{\rm cor}/2.
\end{equation}
Thus, we can define an upper bound $e_{\rm qber}$ of the bit error rate as in Eq.~\eqref{eq:upper_bit_error_rate}, which holds with probability no smaller than $1-\varepsilon_{\rm cor}/2$.

\subsection{Azuma's inequality combined with the Doob decomposition} \label{sec:azuma_doob}
This section summarizes the Azuma's inequality \cite{Azuma_ineq,refined_Azuma_ineq1,refined_Azuma_ineq2} adapted to our security analysis.  The main claim is the following.
\begin{prop}\label{prop:azuma_doob}
    Let $\{\hat{X}^{(i)}\}_{i=1,\ldots,N}$ be a sequence of random variables and $\hat{X}^{(<i)}$ be $\{\hat{X}^{(k)}\}_{k=0,\ldots,i-1}$, where $\hat{X}^{(0)}$ is defined to be zero.
    Let $\{a_i\}_{i=1,\ldots,N}$ and $\{b_i\}_{i=1,\ldots,N}$ be predictable processes with respect to $\{\hat{X}^{(k)}\}_{k=0,\ldots,N}$, i.e., $a_i$ and $b_i$ are constants conditioned on $\hat{X}^{(<i)}$.  Assume the following condition is satisfied for all $i\in\{1,\ldots,N\}$:
    \begin{equation}
        a_i \leq \hat{X}^{(i)}\leq b_i.
    \end{equation}
    Then, for any $t\geq 0$, we have
    \begin{equation}
        \mathrm{Pr}\left(\sum_{i=1}^N \bigl(\hat{X}^{(i)} - \mathbb{E}[\hat{X}^{(i)}\mid \hat{X}^{(<i)}]\bigr) \geq t\right) \leq \exp\biggl(-\frac{2t^2}{\sum_{i=1}^N (b_i - a_i)^2}\biggr). 
    \end{equation}
\end{prop}
\begin{proof}
        For $1\leq i\leq N$, let $\hat{Y}^{(i)}$ be defined as 
        \begin{equation}
            \hat{Y}^{(i)}\coloneqq \sum_{k=1}^{i}(\hat{X}^{(k)}- \mathbb{E}[\hat{X}^{(k)}\mid \hat{X}^{(<k)}]).
        \end{equation}
        Then, the sequence $\{\hat{Y}^{(i)}\}_{i=1,\ldots,N}$ is a martingale with respect to $\{\hat{X}^{(k)}\}_{k=0,\ldots,N}$, since $\hat{Y}^{(i)}$ satisfies 
        \begin{align}
            \mathbb{E}[\hat{Y}^{(i)}\mid \hat{X}^{(<i)}] &= \sum_{k=1}^{i}\left(\mathbb{E}[\hat{X}^{(k)}\mid \hat{X}^{(<i)}] - \mathbb{E}\!\left[\mathbb{E}[\hat{X}^{(k)}\mid \hat{X}^{(<k)}]\relmiddle|\hat{X}^{(<i)}\right]\right) \\
            &= \sum_{k=1}^{i-1}(\hat{X}^{(k)} - \mathbb{E}[\hat{X}^{(k)}\mid \hat{X}^{(<k)}]) \\
            &= \hat{Y}^{(i-1)}
        \end{align}
        for $1\leq i \leq N$.  (Note that $\hat{Y}^{(0)}\coloneqq 0$.)
        The second equality follows from the fact that $\mathbb{E}[\hat{X}^{(k)}\mid \hat{X}^{(<i)}]=\hat{X}^{(k)}$ holds when $k<i$.  Furthermore, the sequence $\{\hat{Y}^{(i)}\}_{i=1,\ldots,N}$ satisfies
        \begin{equation}
            a_i - \mathbb{E}[\hat{X}^{(i)}\mid \hat{X}^{(<i)}] \leq \hat{Y}^{(i)} - \hat{Y}^{(i-1)} \leq b_i - \mathbb{E}[\hat{X}^{(i)}\mid \hat{X}^{(<i)}], 
        \end{equation}
        for $1\leq i\leq n$.  Since $a_i- \mathbb{E}[\hat{X}^{(i)}\mid \hat{X}^{(<i)}]$ and $b_i- \mathbb{E}[\hat{X}^{(i)}\mid \hat{X}^{(<i)}]$ are predictable process with respect to $\{\hat{X}^{(k)}\}_{k=0,\ldots,N}$, we can apply Azuma's inequality~\cite{Azuma_ineq} to the sequence $\{\hat{Y}^{(i)}\}_{i=1,\ldots,N}$ and prove the statement.
\end{proof}

\section{Proof of the operator inequality} \label{sec:operator_ineq}
In this section, we prove the inequality \eqref{eq:ope_ineq_first} used in the security proof in the main text.  We first prove the following lemma.

\begin{lemma} \label{lemma:operator_ineq}
    Let $\Pi_{\pm}$ be orthogonal projections that have the rank no smaller than three or infinite.  Let $M$ be a self-adjoint operator satisfying $M = (\Pi_{+} + \Pi_{-}) M (\Pi_{+} + \Pi_{-}) \leq \alpha (\Pi_{+} + \Pi_{-})$, where $\alpha$ is a real constant.  Let $\ket{\psi}$ be a vector satisfying $(\Pi_{+} + \Pi_{-})\ket{\psi} = \ket{\psi}$ and $\Pi_{\pm}\ket{\psi} \neq 0$.  Assume $\Pi_{\pm}\ket{\psi}$ are not proportional to eigenvectors of $\Pi_{\pm}M\Pi_{+}$ (if they have).  Define the following quantities with respect to $\ket{\psi}$:
    \begin{align}
        C_{\pm} &\coloneqq \bra{\psi}\Pi_{\pm}\ket{\psi} \, (>0), \label{eq:def_c_pm}\\
        \lambda_{++} &\coloneqq C_{+}^{-1}\bra{\psi} M_{++}\ket{\psi}, \quad \lambda_{--} \coloneqq C_{-}^{-1}\bra{\psi} M_{--}\ket{\psi}, \label{eq:lam_pm}\\
        \lambda_{+-} &\coloneqq (C_{+}C_{-})^{-\frac{1}{2}}\bra{\psi} M_{+-} \ket{\psi}, \quad \lambda_{-+} \coloneqq \lambda_{+-}^{*}, \\
        \sigma_{\pm +} &\coloneqq \(C_{+}^{-1}\|M_{\pm +}\ket{\psi}\|^2 - |\lambda_{\pm +}|^{2}\)^{\frac{1}{2}}, \label{eq:def_sig_pm}\\
        \sigma_{\pm -} &\coloneqq \sigma_{\pm +}^{-1} \((C_+ C_-)^{-\frac{1}{2}} \bra{\psi}M_{+ \pm} M_{\pm -} \ket{\psi} - \lambda_{+-}\lambda_{\pm\pm}\), \\
        \Delta_{\pm -} & \coloneqq \(C_{-}^{-1}\|M_{\pm -}\ket{\psi}\|^{2} - |\lambda_{\pm -}|^2 - |\sigma_{\pm -}|^2\)^{\frac{1}{2}} \label{eq:def_delta_pm},
    \end{align}
    where $M_{++}, M_{--}, M_{+-},$ and $M_{-+}$ are given respectively by
    \begin{equation}
        M_{\pm \pm} \coloneqq \Pi_{\pm} M \Pi_{\pm}, \quad M_{+ -} \coloneqq \Pi_{+} M \Pi_{-}, \quad M_{-+}\coloneqq M_{+-}^{\dagger}. \label{eq:def_M_pm}
    \end{equation}
    Then, for any real numbers $\gamma_{\pm}$, we have
    \begin{equation}
        \sigma_{\rm sup}(M + \ket{\psi}\!\bra{\psi} - \gamma_{+}\Pi_{+} - \gamma_{-}\Pi_{-}) \leq \sigma_{\rm sup}(M_{\rm 6d}), \label{eq:operator_ineq}
    \end{equation}
    where $\sigma_{\rm sup}(X)$ denotes the supremum of the spectrum of the operator $X$, and $M_{\rm 6d}$ is given by
    \begin{equation}
        M_{\rm 6d} \coloneqq \begin{pmatrix} \alpha - \gamma_+ & 0 & 0 & \Delta_{+-} & 0 & 0 \\
            0 & \alpha - \gamma_+ & \sigma_{++} & \sigma_{+-} & 0 & 0 \\
            0 & \sigma_{++} & C_+ + \lambda_{++} - \gamma_+ & \sqrt{C_+ C_-} + \lambda_{+-} & \sigma_{-+} & 0 \\
            \Delta_{+-} & \sigma_{+-}^{*} & \sqrt{C_+ C_-} + \lambda_{-+} & C_- + \lambda_{--} - \gamma_- & \sigma_{--}^{*} & \Delta_{--} \\
            0 & 0 & \sigma_{-+} & \sigma_{--} & \alpha - \gamma_{-} & 0 \\
            0 & 0 & 0 & \Delta_{--} & 0 & \alpha - \gamma_-
        \end{pmatrix}. \label{eq:def_of_M_6d}
    \end{equation}
\end{lemma}

\begin{proof}
    We are able to choose orthonormal vectors $\{\ket{e^{(1)}_{+}}, \ket{e^{(2)}_{+}}, \ket{e^{(3)}_{+}}\}$ and $\{\ket{e^{(1)}_{-}}, \ket{e^{(2)}_{-}}, \ket{e^{(3)}_{-}}\}$ in the domains of $\Pi_{+}$ and $\Pi_{-}$, respectively, to satisfy
    \begin{align}
        \sqrt{C_{\pm}}\Ket{e^{(1)}_{\pm}} &= \Pi_{\pm}\ket{\psi}, \label{eq:equality_c_pm}\\
        M\Ket{e^{(1)}_{+}} = (M_{++} + M_{-+}) \Ket{e^{(1)}_{+}} &= \lambda_{++}\Ket{e^{(1)}_{+}} + \sigma_{++}\Ket{e^{(2)}_{+}} + \lambda_{-+}\Ket{e^{(1)}_{-}} + \sigma_{-+}\Ket{e^{(2)}_{-}}, \label{eq:m_+_expansion} \\
        \begin{split}
        M\Ket{e^{(1)}_{-}} = (M_{+-} + M_{--}) \Ket{e^{(1)}_{-}} &= \lambda_{+-} \Ket{e^{(1)}_{+}} + \sigma_{+-} \Ket{e^{(2)}_{+}} + \Delta_{+-}\Ket{e^{(3)}_{+}} \\
        & \qquad + \lambda_{--}\Ket{e^{(1)}_{-}} + \sigma_{--}\Ket{e^{(2)}_{-}} + \Delta_{--} \Ket{e^{(3)}_{-}}. \end{split}\label{eq:m_-_expansion}
    \end{align}
    The coefficients satisfy Eqs.~\eqref{eq:def_c_pm}--\eqref{eq:def_M_pm}, and $M = (\Pi_{+} + \Pi_{-}) M (\Pi_{+} + \Pi_{-})$.  This can be seen by applying the following general property.  For an operator $A$ and a unit vector $\ket{v}$, $A\ket{v}$ can always be written as $A\ket{v}=a\ket{v}+b\ket{v^{\perp}}$, where $\ket{v^{\perp}}$ is a unit vector orthogonal to $\ket{v}$, $b$ is a nonnegative coefficient satisfying $b^2=\|A\ket{v}\|^2 - |a|^2$, and $a=\bra{v}A\ket{v}$.
    From $(\Pi_{+} + \Pi_{-})\ket{\psi} = \ket{\psi}$, we have
    \begin{equation}
        \ket{\psi} = \sqrt{C_+} \Ket{e^{(1)}_{+}} + \sqrt{C_-} \Ket{e^{(1)}_{-}}.
    \end{equation}
    Let us now define the following projection operators:
    \begin{align}
        \Pi_{\pm}^{(j)} &\coloneqq \Ket{e^{(j)}_{\pm}}\! \Bra{e^{(j)}_{\pm}} \quad (j=1,2,3), \\
        \Pi_{\pm}^{(\geq 2)} &\coloneqq \Pi_{\pm} - \Pi_{\pm}^{(1)}, \\
        \Pi_{\pm}^{(\geq 4)} &\coloneqq \Pi_{\pm}^{(\geq 2)} - \Pi_{\pm}^{(2)} - \Pi_{\pm}^{(3)}.
    \end{align}
    Since Eqs.~\eqref{eq:m_+_expansion} and \eqref{eq:m_-_expansion} imply $(\Pi_{+}^{(\geq 4)} + \Pi_{-}^{(\geq 4)}) M (\Pi_{+}^{(1)} + \Pi_{-}^{(1)}) = 0$, we have
    \begin{align}
        M &= (\Pi_{+} + \Pi_{-}) M (\Pi_{+} + \Pi_{-}) \\
        \begin{split}
        &=  (\Pi_{+}^{(1)} + \Pi_{-}^{(1)}) M (\Pi_{+}^{(1)} + \Pi_{-}^{(1)}) + (\Pi_{+}^{(2)} + \Pi_{+}^{(3)} + \Pi_{-}^{(2)} + \Pi_{-}^{(3)}) M (\Pi_{+}^{(1)} + \Pi_{-}^{(1)}) \\
        & \qquad + (\Pi_{+}^{(1)} + \Pi_{-}^{(1)}) M (\Pi_{+}^{(2)} + \Pi_{+}^{(3)} + \Pi_{-}^{(2)} + \Pi_{-}^{(3)}) + (\Pi_{+}^{(\geq 2)} + \Pi_{-}^{(\geq 2)}) M (\Pi_{+}^{(\geq 2)} + \Pi_{-}^{(\geq 2)}) \end{split} \\
        \begin{split}
        &\leq \lambda_{++}\Pi_{+}^{(1)} + \lambda_{--}\Pi_{-}^{(1)} + \lambda_{+-}\Ket{e_{+}^{(1)}}\!\Bra{e_{-}^{(1)}} + \lambda_{-+}\Ket{e_{-}^{(1)}}\!\Bra{e_{+}^{(1)}} \\
        &\qquad +\(\sigma_{++}\Ket{e_{+}^{(2)}}\!\Bra{e_{+}^{(1)}} + \sigma_{-+}\Ket{e_{-}^{(2)}}\!\Bra{e_{+}^{(1)}} + \sigma_{+-}\Ket{e_{+}^{(2)}}\!\Bra{e_{-}^{(1)}} \right.\\ 
        & \qquad \qquad \left. + \Delta_{+-}\Ket{e_{+}^{(3)}}\!\Bra{e_{-}^{(1)}} + \sigma_{--}\Ket{e_{-}^{(2)}}\!\Bra{e_{-}^{(1)}} + \Delta_{--}\Ket{e_{-}^{(3)}}\!\Bra{e_{-}^{(1)}} \) + (\quad \text{h.c.}\quad ) \\
        & \qquad \ \  + \alpha (\Pi_{+}^{(\geq 2)} + \Pi_{-}^{(\geq 2)}),
        \end{split} \label{eq:ineq_of_M}
    \end{align}
    where h.c.\ denotes the hermitian conjugate of the terms in the preceding parenthesis.
    The last inequality comes from $M\leq \alpha(\Pi_{+} + \Pi_{-})$.  Using Eq.~\eqref{eq:ineq_of_M}, we have
    \begin{equation}
        M + \ket{\psi}\!\bra{\psi} - \gamma_{+}\Pi_{+} - \gamma_{-}\Pi_{-} \leq M_{\rm 6d} \oplus (\alpha - \gamma_+)\Pi_{+}^{(\geq 4)} \oplus (\alpha - \gamma_{-})\Pi_{-}^{(\geq 4)}, \label{eq:bound_by_6d}
    \end{equation}
    where $M_{\rm 6d}$ is given in Eq.~\eqref{eq:def_of_M_6d} with the basis $\{\ket{e_{+}^{(3)}}, \ket{e_{+}^{(2)}}, \ket{e_{+}^{(1)}}, \ket{e_{-}^{(1)}}, \ket{e_{-}^{(2)}}, \ket{e_{-}^{(3)}}\}$.  Since $\alpha-\gamma_{\pm} = \bra{e_{\pm}^{(3)}}M_{\rm 6d}\ket{e_{\pm}^{(3)}} \leq \sigma_{\rm sup}(M_{\rm 6d})$, the supremum of the spectrum of the right-hand side of Eq.~\eqref{eq:bound_by_6d} is equal to the maximum eigenvalue of the six-dimensional matrix $M_{\rm 6d}$.  We then obtain Eq.~\eqref{eq:operator_ineq}.
\end{proof}

As a corollary of this lemma, we obtain the followings.  
First, we consider Homodyne protocol.  
\begin{cor}\label{cor:ope_ineq_M}
    Let $\ket{\beta}$ be a coherent state and $\theta_{\mu,\beta}^{\rm hom}(x)$ be defined to satisfy
    \begin{equation}
        |\theta_{\mu,\beta}^{\rm hom}(x)| \leq \frac{\pi}{2}, \qquad \tan\theta_{\mu,\beta}^{\rm hom}(x) = e^{-2(\mu - \beta^2)}\sinh(4\beta x).\label{eq:def_of_theta}
    \end{equation}
    Let $\Pi_{\rm ev(od)}$ and $M^{\rm hom}[\kappa,\gamma]$ be as defined in the main text, and let $M^{\rm hom}_{\rm oo}$, $M^{\rm hom}_{\rm ee}$, $M^{\rm hom}_{\rm (\pm,o)(\mp,e)}$, and $M^{\rm hom}_{\rm (\mp,e)(\pm,o)}$ be defined as follows:
    \begin{align}
        M^{\rm hom}_{\rm oo} &\coloneqq \int_{-\infty}^{\infty}f_{\rm suc}(x) [1 + \cos\theta_{\mu,\beta}^{\rm hom}(x)] dx\; \Pi_{\rm od} \ket{x}\!\bra{x} \Pi_{\rm od}, \\
        M^{\rm hom}_{\rm ee} &\coloneqq \int_{-\infty}^{\infty}f_{\rm suc}(x) [1 - \cos\theta_{\mu,\beta}^{\rm hom}(x)]dx\; \Pi_{\rm ev} \ket{x}\!\bra{x} \Pi_{\rm ev}, \\  
        M^{\rm hom}_{\rm (+,o)(-,e)} &\coloneqq \int_{-\infty}^{\infty}f_{\rm suc}(x)\sin\theta_{\mu,\beta}^{\rm hom}(x)\,dx\; \Pi_{\rm od} \ket{x}\!\bra{x} \Pi_{\rm ev}, \\
        M^{\rm hom}_{\rm (-,e)(+,o)} &\coloneqq \left(M^{\rm hom}_{\rm (+,o)(-,e)}\right)^{\dagger},\\
        M^{\rm hom}_{\rm (-,o)(+,e)} &\coloneqq \int_{-\infty}^{\infty}-f_{\rm suc}(x)\,\sin\theta_{\mu,\beta}^{\rm hom}(x)\,dx\; \Pi_{\rm od} \ket{x}\!\bra{x} \Pi_{\rm ev}=-M^{\rm hom}_{\rm (+,o)(-,e)}, \\
        M^{\rm hom}_{\rm (+,e)(-,o)} &\coloneqq \left(M^{\rm hom}_{\rm (-,o)(+,e)}\right)^{\dagger}=-\left(M^{\rm hom}_{\rm (+,o)(-,e)}\right)^{\dagger}.
    \end{align}  
    Define the following (real) parameters:
    \begin{gather}
        C_{\rm o} \coloneqq \bra{\beta}\Pi_{\rm od}\ket{\beta} = e^{-|\beta|^2} \sinh|\beta|^2 , \quad 
        C_{\rm e} \coloneqq \bra{\beta}\Pi_{\rm ev}\ket{\beta} = e^{-|\beta|^2}\cosh |\beta|^2, 
         \\
        \lambda_{\rm oo}^{\rm hom} \coloneqq C_{\rm o}^{-1}\bra{\beta} M^{\rm hom}_{\rm oo} \ket{\beta} , \quad
        \lambda_{\rm ee}^{\rm hom} \coloneqq C_{\rm e}^{-1} \bra{\beta} M^{\rm hom}_{\rm ee} \ket{\beta},  \\
        \lambda_{\rm (+,o)(-,e)}^{\rm hom} \coloneqq \left(C_{\rm o}C_{\rm e}\right)^{-\frac{1}{2}} \bra{\beta} M^{\rm hom}_{\rm (+,o)(-,e)} \ket{\beta} = (\lambda_{\rm (+,o)(-,e)}^{\rm hom})^*, \\
        \sigma_{\rm oo}^{\rm hom} \coloneqq \( C_{\rm o}^{-1} \| M^{\rm hom}_{\rm oo}\ket{\beta} \|^2 - (\lambda_{\rm oo}^{\rm hom})^2 \)^{\frac{1}{2}}, \\
        \sigma_{\rm eo}^{\rm hom} \coloneqq \( C_{\rm o}^{-1} \| M^{\rm hom}_{\rm (-,e)(+,o)}\ket{\beta} \|^2 - |\lambda_{\rm (+,o)(-,e)}^{\rm hom}|^2 \)^{\frac{1}{2}}, \\
        \sigma_{\rm (+,o)(-,e)}^{\rm hom} \coloneqq (\sigma_{\rm oo}^{\rm hom})^{-1} \(\left(C_{\rm o} C_{\rm e}\right)^{-\frac{1}{2}} \bra{\beta} M^{\rm hom}_{\rm oo} M^{\rm hom}_{\rm (+,o)(-,e)}\ket{\beta} - \lambda_{\rm oo}^{\rm hom}\lambda_{\rm (+,o)(-,e)}^{\rm hom}\) = (\sigma_{\rm (+,o)(-,e)}^{\rm hom})^* , \\
        \sigma_{\rm (-,e)(-,e)}^{\rm hom} \coloneqq (\sigma_{\rm (-,e)(+,o)}^{\rm hom})^{-1} \(\left(C_{\rm o} C_{\rm e}\right)^{-\frac{1}{2}} \bra{\beta} M^{\rm hom}_{\rm (+,o)(-,e)} M^{\rm hom}_{\rm ee}\ket{\beta} - \lambda_{\rm (+,o)(-,e)}^{\rm hom}\lambda_{\rm ee}^{\rm hom}\) = (\sigma_{\rm (-,e)(-,e)}^{\rm hom})^*,\\
        \Delta_{\rm oe}^{\rm hom} \coloneqq \(C_{\rm e}^{-1}\|M^{\rm hom}_{\rm (+,o)(-,e)}\ket{\beta}\|^2 - |\lambda_{\rm (+,o)(-,e)}^{\rm hom}|^2 - |\sigma_{\rm (+,o)(-,e)}^{\rm hom}|^2 \)^{\frac{1}{2}}, \\
        \Delta_{\rm ee}^{\rm hom} \coloneqq \(C_{\rm e}^{-1} \|M^{\rm hom}_{\rm ee}\ket{\beta}\|^2 - (\lambda_{\rm ee}^{\rm hom})^2 - |\sigma_{\rm (-,e)(-,e)}^{\rm hom}|^2 \)^{\frac{1}{2}}.
    \end{gather}
    Define the following two matrices $M_{\rm 6d}^{(0)}$ and $M_{\rm 6d}^{(1)}$.
    \begin{align}
        M_{\rm 6d}^{(0)} &\coloneqq 
        \begin{pmatrix}
            1 & 0 & 0 & \Delta_{\rm oe}^{\rm hom} & 0 & 0 \\
            0 & 1 & \sigma_{\rm oo}^{\rm hom} & \sigma_{\rm (+,o)(-,e)}^{\rm hom} & 0 & 0 \\
            0 & \sigma_{\rm oo}^{\rm hom} & \kappa C_{\rm o} + \lambda_{\rm oo}^{\rm hom} &  \kappa \sqrt{C_{\rm o}C_{\rm e}} + \lambda_{\rm (+,o)(-,e)}^{\rm hom} & \sigma_{\rm eo}^{\rm hom} & 0 \\
            \Delta_{\rm oe}^{\rm hom} & \sigma_{\rm (+,o)(-,e)}^{\rm hom} & \kappa \sqrt{C_{\rm o}C_{\rm e}} + \lambda_{\rm (+,o)(-,e)}^{\rm hom} & \kappa C_{\rm e} + \lambda_{\rm ee}^{\rm hom} - \gamma & \sigma_{\rm (-,e)(-,e)}^{\rm hom} & \Delta_{\rm ee}^{\rm hom} \\
            0 & 0 & \sigma_{\rm eo}^{\rm hom} & \sigma_{\rm (-,e)(-,e)}^{\rm hom} & 1 - \gamma & 0 \\
            0 & 0 & 0 & \Delta_{\rm ee}^{\rm hom} & 0 & 1 - \gamma
        \end{pmatrix}, \label{eq:M_6d_0}\\[1ex]
        M_{\rm 6d}^{(1)} &\coloneqq 
        \begin{pmatrix}
            1 - \gamma & 0 & 0 & \Delta_{\rm oe}^{\rm hom} & 0 & 0 \\
            0 & 1 - \gamma & \sigma_{\rm oo}^{\rm hom} & -\sigma_{\rm (+,o)(-,e)}^{\rm hom} & 0 & 0 \\
            0 & \sigma_{\rm oo}^{\rm hom} & \kappa C_{\rm o} + \lambda_{\rm oo}^{\rm hom} - \gamma & \kappa \sqrt{C_{\rm o}C_{\rm e}} - \lambda_{\rm (+,o)(-,e)}^{\rm hom} & \sigma_{\rm eo}^{\rm hom} & 0 \\
            \Delta_{\rm oe}^{\rm hom} & -\sigma_{\rm (+,o)(-,e)}^{\rm hom} & \kappa \sqrt{C_{\rm o}C_{\rm e}} - \lambda_{\rm (+,o)(-,e)}^{\rm hom} & \kappa C_{\rm e} + \lambda_{\rm ee}^{\rm hom} & -\sigma_{\rm (-,e)(-,e)}^{\rm hom} & \Delta_{\rm ee}^{\rm hom} \\
            0 & 0 & \sigma_{\rm eo}^{\rm hom} & -\sigma_{\rm (-,e)(-,e)}^{\rm hom} & 1 & 0 \\
            0 & 0 & 0 & \Delta_{\rm ee}^{\rm hom} & 0 & 1 
        \end{pmatrix}.\label{eq:M_6d_1}
    \end{align}
    Define a convex function
    \begin{equation}
        B^{\rm hom}(\kappa,\gamma) \coloneqq \max\{\sigma_{\rm sup}(M_{\rm 6d}^{(0)}), \sigma_{\rm sup}(M_{\rm 6d}^{(1)})\}. \label{eq:B_kappa_gamma_hom}
    \end{equation}
    Then, for $\kappa,\gamma\geq 0$, we have
    \begin{equation}
        M^{\rm hom}[\kappa,\gamma] \leq B^{\rm hom}(\kappa,\gamma) I_{AC}.
        \label{eq:bounded_by_B}
    \end{equation}
\end{cor}

\begin{proof}
    We first derive the explicit form of $\ket{u_\pm^{\rm hom}(x)}_A$ introduced in Eq.~\eqref{eq:W_A_hom}. Notice that
    \begin{align}
        &1 - 2q_{\mu,\beta} = e^{-2(\mu -\beta^2)}, \label{eq:one_minus_2q}\\
        &\sqrt{\frac{g_{\beta,1/4}(x)}{g_{-\beta,1/4}(x)}} = e^{4\beta x}. \label{eq:gaussian_ratio}
    \end{align}
    Let $\theta(x)$ be defined to satisfy
    \begin{equation}
        |\theta(x)| < \frac{\pi}{2}, \qquad \tan\theta(x) = \mathrm{Tr}\left(Z_A \bra{0}_B\tau_{AB}^{\rm hom}(x)\ket{1}_B \right) \Bigl/ \mathrm{Tr}\left(X_A \bra{0}_B\tau_{AB}^{\rm hom}(x)\ket{1}_B \right). \label{eq:theta_def_alt}
    \end{equation}
    Noticing that $\mathrm{Tr}\left(Y_A \bra{0}_B\tau_{AB}^{\rm hom}(x)\ket{1}_B \right)=0$, we have 
    \begin{equation}
        \ket{u_\pm^{\rm hom}(x)}_A = \cos\frac{\theta(x)}{2}\ket{\pm}_A \pm \sin\frac{\theta(x)}{2} \ket{\mp}_A. \label{eq:form_of_u_pm}
    \end{equation}
    From Eqs.~\eqref{eq:tau_hom_0_1}, \eqref{eq:one_minus_2q}, \eqref{eq:gaussian_ratio}, and \eqref{eq:theta_def_alt}, we can see that $\theta(x)$ coincides with $\theta_{\mu,\beta}^{\rm hom}(x)$ defined in Eq.~\eqref{eq:def_of_theta}.
    We now observe that
    \begin{align}
        &|\braket{+|u_+^{\rm hom}(x)}|^2 = |\braket{-|u_-^{\rm hom}(x)}|^2 = \cos^2\biggl(\frac{\theta_{\mu,\beta}^{\rm hom}(x)}{2}\biggr) = \frac{1+\cos\theta_{\mu,\beta}^{\rm hom}(x)}{2}, \label{eq:sandwiched_hom_od_od}\\
        &|\braket{-|u_+^{\rm hom}(x)}|^2 = |\braket{+|u_-^{\rm hom}(x)}|^2 = \sin^2\biggl(\frac{\theta_{\mu,\beta}^{\rm hom}(x)}{2}\biggr) = \frac{1-\cos\theta_{\mu,\beta}^{\rm hom}(x)}{2} , \label{eq:sandwiched_hom_ev_ev}\\
        \begin{split}
            &\braket{+|u_+^{\rm hom}(x)}\!\braket{u_+^{\rm hom}(x)|-} = \braket{-|u_+^{\rm hom}(x)}\!\braket{u_+^{\rm hom}(x)|+} = \sin\biggl(\frac{\theta_{\mu,\beta}^{\rm hom}(x)}{2}\biggr)\cos\biggl(\frac{\theta_{\mu,\beta}^{\rm hom}(x)}{2}\biggr) \\
            &= -\braket{+|u_-^{\rm hom}(x)}\!\braket{u_-^{\rm hom}(x)|-} = -\braket{-|u_-^{\rm hom}(x)}\!\braket{u_-^{\rm hom}(x)|+} = \frac{\sin\theta_{\mu,\beta}^{\rm hom}(x)}{2}.
        \end{split} \label{eq:sandwiched_hom_plus_od_minus_ev}
    \end{align}
    From Eq.~\eqref{eq:hom_direct_sum} as well as Eqs.~\eqref{eq:def_of_pi_fid}--\eqref{eq:def_of_M_kappa_gamma}, it is obvious that
    \begin{equation}
        M^{\rm hom}[\kappa,\gamma] = \Pi_{AC}^{(+,\mathrm{od}), (-,\mathrm{ev})} M^{\rm hom}[\kappa,\gamma]\, \Pi_{AC}^{(+,\mathrm{od}), (-,\mathrm{ev})} + \Pi_{AC}^{(-,\mathrm{od}),(+,\mathrm{ev})} M^{\rm hom}[\kappa,\gamma]\, \Pi_{AC}^{(-,\mathrm{od}), (+,\mathrm{ev})},
    \end{equation}
    where the two orthogonal projections $\Pi_{AC}^{(+,\mathrm{od}),(-,\mathrm{ev})}$ and $\Pi_{AC}^{(-,\mathrm{od}), (+,\mathrm{ev})}$ are defined in Eqs.~\eqref{eq:plus_od_minus_ev} and \eqref{eq:minus_od_plus_ev}.
    Then we apply Lemma \ref{lemma:operator_ineq} respectively to the operators $\Pi_{AC}^{(+,\mathrm{od}), (-,\mathrm{ev})} M^{\rm hom}[\kappa,\gamma]\, \Pi_{AC}^{(+,\mathrm{od}), (-,\mathrm{ev})}$ and $\Pi_{AC}^{(-,\mathrm{od}),(+,\mathrm{ev})} M^{\rm hom}[\kappa,\gamma]\, \Pi_{AC}^{(-,\mathrm{od}), (+,\mathrm{ev})}$.  For $\Pi_{AC}^{(+,\mathrm{od}), (-,\mathrm{ev})} M^{\rm hom}[\kappa,\gamma]\, \Pi_{AC}^{(+,\mathrm{od}), (-,\mathrm{ev})}$, we set, by using Eqs.~\eqref{eq:sandwiched_hom_od_od}--\eqref{eq:sandwiched_hom_plus_od_minus_ev}, that
    \begin{align}
        \Pi_{\pm} &= \ket{\pm}\!\bra{\pm}_A \otimes \Pi_{\rm od(ev)}, \\
        M &= \Pi_{AC}^{(+,\mathrm{od}), (-,\mathrm{ev})} M^{\rm hom}_{\rm ph} \Pi_{AC}^{(+,\mathrm{od}), (-,\mathrm{ev})} \\
        &= \ket{+}\!\bra{+}_A\otimes M^{\rm hom}_{\rm oo} + \ket{-}\!\bra{-}_A\otimes M^{\rm hom}_{\rm ee} + \left( \ket{+}\!\bra{-}_A\otimes M^{\rm hom}_{\rm (+,o)(-,e)} + \ket{-}\!\bra{+}_A\otimes M^{\rm hom}_{\rm (-,e)(+,o)} \right), \\
        \ket{\psi} &= \sqrt{\kappa} \ket{\phi_-}_{AC}, \\
        \alpha &= 1, \quad \gamma_+ = 0, \quad \gamma_{-} = \gamma,
    \end{align}
    where $\ket{\phi_-}_{AC}$ is defined in Eq.~\eqref{eq:phi_minus}.  Since so-defined $M$ only has continuous spectrum, we can apply Lemma~\ref{lemma:operator_ineq} and obtain
    \begin{equation}
        \sigma_{\rm sup}\!\(\Pi_{AC}^{(+,\mathrm{od}), (-,\mathrm{ev})} M^{\rm hom}[\kappa,\gamma]\, \Pi_{AC}^{(+,\mathrm{od}), (-,\mathrm{ev})}\) \leq \sigma_{\rm sup}(M_{\rm 6d}^{(0)}). \label{eq:bounded_by_6d_0}
    \end{equation}
    In the same way, we apply Lemma \ref{lemma:operator_ineq} to $\Pi_{AC}^{(-,\mathrm{od}),(+,\mathrm{ev})} M^{\rm hom}[\kappa,\gamma]\, \Pi_{AC}^{(-,\mathrm{od}), (+,\mathrm{ev})}$.  Using Eqs.~\eqref{eq:sandwiched_hom_od_od}--\eqref{eq:sandwiched_hom_plus_od_minus_ev}, we set
    \begin{align}
        \Pi_{\pm} &= \ket{\mp}\!\bra{\mp}_A \otimes \Pi_{\rm od(ev)}, \\
        M &= \Pi_{AC}^{(-,\mathrm{od}), (+,\mathrm{ev})} M^{\rm hom}_{\rm ph} \Pi_{AC}^{(-,\mathrm{od}), (+,\mathrm{ev})} \\
        &= \ket{-}\!\bra{-}_A\otimes M^{\rm hom}_{\rm oo} + \ket{+}\!\bra{+}_A\otimes M^{\rm hom}_{\rm ee} + \left( \ket{-}\!\bra{+}_A\otimes M^{\rm hom}_{\rm (-,o)(+,e)} + \ket{+}\!\bra{-}_A\otimes M^{\rm hom}_{\rm (+,e)(-,o)} \right), \\
        &=\ket{-}\!\bra{-}_A\otimes M^{\rm hom}_{\rm oo} + \ket{+}\!\bra{+}_A\otimes M^{\rm hom}_{\rm ee} - \left( \ket{-}\!\bra{+}_A\otimes M^{\rm hom}_{\rm (+,o)(-,e)} + \ket{+}\!\bra{-}_A\otimes M^{\rm hom}_{\rm (-,e)(+,o)} \right), \\
        \ket{\psi} &= \sqrt{\kappa} \ket{\phi_+}_{AC}, \\
        \alpha &= 1, \quad \gamma_+ = \gamma, \quad \gamma_{-} = 0,
    \end{align}
    where $\ket{\phi_+}_{AC}$ is defined in Eq.~\eqref{eq:phi_plus}.
    Then, we observe 
    \begin{equation}
        \sigma_{\rm sup}\!\(\Pi_{AC}^{(-,\mathrm{od}),(+,\mathrm{ev})} M^{\rm hom}[\kappa,\gamma]\, \Pi_{AC}^{(-,\mathrm{od}), (+,\mathrm{ev})}\) \leq \sigma_{\rm sup}(M_{\rm 6d}^{(1)}).
        \label{eq:bounded_by_6d_1}
    \end{equation}
    Combining inequalities \eqref{eq:bounded_by_6d_0} and \eqref{eq:bounded_by_6d_1} completes the proof.
\end{proof}

Next, we consider Heterodyne protocol.
\begin{cor}\label{cor:ope_ineq_M_het}
    Let $\ket{\beta}$ be a coherent state and $\theta_{\mu,\beta}^{\rm hom}(x)$ be defined to satisfy
    \begin{equation}
        |\theta_{\mu,\beta}^{\rm het}(\omega_r)| \leq \frac{\pi}{2}, \qquad \tan\theta_{\mu,\beta}^{\rm het}(\omega_r) = e^{-2(\mu - \beta^2)}\sinh(2\beta \omega_r).\label{eq:def_of_theta_het}
    \end{equation}  
    Let $\Pi_{\rm ev(od)}$ and $M^{\rm het}[\kappa,\gamma]$ be as defined in the main text, and let $M^{\rm het}_{\rm oo}$, $M^{\rm het}_{\rm ee}$, $M^{\rm het}_{\rm (\pm,o)(\mp,e)}$, and $M^{\rm het}_{\rm (\mp,e)(\pm,o)}$ be defined as follows:
    \begin{align}
        \begin{split}
        M^{\rm het}_{\rm oo} &\coloneqq \iint_{-\infty}^{\infty} f_{\rm suc}(\omega_r)d\omega_r dx \;\Pi_{\rm od} \Bigl[\ket{x}\!\braket{x|\omega_r}\!\braket{\omega_r|x}\!\bra{x} \\
        & \hspace{1.5cm}  +\frac{\cos\theta_{\mu,\beta}^{\rm het}(\omega_r)}{2}\bigl(\ket{x}\!\braket{x|\omega_r}\!\braket{\omega_r|x-\beta}\!\bra{x-\beta} + \ket{x-\beta}\!\braket{x-\beta|\omega_r}\!\braket{\omega_r|x}\!\bra{x}\bigr)\Bigr] \Pi_{\rm od}, \end{split} \\
        \begin{split}
        M^{\rm het}_{\rm ee} &\coloneqq \iint_{-\infty}^{\infty} f_{\rm suc}(\omega_r)d\omega_r dx \;\Pi_{\rm ev} \Bigl[\ket{x}\!\braket{x|\omega_r}\!\braket{\omega_r|x}\!\bra{x} \\
        & \hspace{1.5cm}  -\frac{\cos\theta_{\mu,\beta}^{\rm het}(\omega_r)}{2}\bigl(\ket{x}\!\braket{x|\omega_r}\!\braket{\omega_r|x-\beta}\!\bra{x-\beta} + \ket{x-\beta}\!\braket{x-\beta|\omega_r}\!\braket{\omega_r|x}\!\bra{x}\bigr)\Bigr] \Pi_{\rm ev}, \end{split}\\
        \begin{split}
        M^{\rm het }_{\rm (+,o)(-,e)} &\coloneqq \iint_{-\infty}^{\infty}f_{\rm suc}(\omega_r) d\omega_r dx\; \Pi_{\rm od} \Bigl[\sin\theta_{\mu,\beta}^{\rm het}(\omega_r)\ket{x}\!\braket{x|\omega_r}\!\braket{\omega_r|x}\!\bra{x} \\
        &\hspace{1.5cm} -\frac{\cos\theta_{\mu,\beta}^{\rm het}(\omega_r)}{2}\bigl(\ket{x}\!\braket{x|\omega_r}\!\braket{\omega_r|x-\beta}\!\bra{x-\beta} - \ket{x-\beta}\!\braket{x-\beta|\omega_r}\!\braket{\omega_r|x}\!\bra{x}\bigr)\Bigr] \Pi_{\rm ev}, \end{split}\\
        M^{\rm het}_{\rm (-,e)(+,o)} &\coloneqq \left(M^{\rm het }_{\rm (+,o)(-,e)}\right)^{\dagger},\\
        \begin{split}
        M^{\rm het }_{\rm (-,o)(+,e)} &\coloneqq \iint_{-\infty}^{\infty}f_{\rm suc}(\omega_r) d\omega_r dx\; \Pi_{\rm od} \Bigl[-\sin\theta_{\mu,\beta}^{\rm het}(\omega_r)\ket{x}\!\braket{x|\omega_r}\!\braket{\omega_r|x}\!\bra{x} \\
        &\hspace{1.5cm} -\frac{\cos\theta_{\mu,\beta}^{\rm het}(\omega_r)}{2}\bigl(\ket{x}\!\braket{x|\omega_r}\!\braket{\omega_r|x-\beta}\!\bra{x-\beta} - \ket{x-\beta}\!\braket{x-\beta|\omega_r}\!\braket{\omega_r|x}\!\bra{x}\bigr)\Bigr] \Pi_{\rm ev}, \end{split} \\
        M^{\rm het }_{\rm (+,e)(-,o)} &\coloneqq \left(M^{\rm het }_{\rm (-,o)(+,e)}\right)^{\dagger},
    \end{align}  
    Define the following parameters:
    \begin{gather}
        C_{\rm o} \coloneqq \bra{\beta}\Pi_{\rm od}\ket{\beta} = e^{-|\beta|^2} \sinh|\beta|^2 , \quad 
        C_{\rm e} \coloneqq \bra{\beta}\Pi_{\rm ev}\ket{\beta} = e^{-|\beta|^2}\cosh |\beta|^2, 
         \\
        \lambda_{\rm oo}^{\rm het} \coloneqq C_{\rm o}^{-1}\bra{\beta} M^{\rm het}_{\rm oo} \ket{\beta}, \quad
        \lambda_{\rm ee}^{\rm het} \coloneqq C_{\rm e}^{-1}\bra{\beta} M^{\rm het}_{\rm ee} \ket{\beta},  \\
        \lambda_{\rm (+,o)(-,e)}^{\rm het} \coloneqq (C_{\rm o}C_{\rm e})^{-\frac{1}{2}}\bra{\beta} M^{\rm het}_{\rm (+,o)(-,e)} \ket{\beta} = (\lambda_{\rm (+,o)(-,e)}^{\rm het})^*, \\
        \lambda_{\rm (-,o)(+,e)}^{\rm het} \coloneqq (C_{\rm o}C_{\rm e})^{-\frac{1}{2}}\bra{\beta} M^{\rm het}_{\rm (-,o)(+,e)} \ket{\beta} = (\lambda_{\rm (-,o)(+,e)}^{\rm het})^*, \\
        \sigma_{\rm oo}^{\rm het} \coloneqq \( C_{\rm o}^{-1} \left\| M^{\rm het}_{\rm oo} \ket{\beta} \right\|^2 - (\lambda_{\rm oo}^{\rm het})^2 \)^{\frac{1}{2}}, \\
        \sigma_{\rm (-,e)(+,o)}^{\rm het} \coloneqq \(C_{\rm o}^{-1} \left\| M^{\rm het}_{\rm (-,e)(+,o)} \ket{\beta} \right\|^2 - |\lambda_{\rm (+,o)(-,e)}^{\rm het}|^2 \)^{\frac{1}{2}}, \\
        \sigma_{\rm (+,e)(-,o)}^{\rm het} \coloneqq \(C_{\rm o}^{-1} \left\| M^{\rm het}_{\rm (+,e)(-,o)} \ket{\beta} \right\|^2 - |\lambda_{\rm (-,o)(+,e)}^{\rm het}|^2 \)^{\frac{1}{2}}, \\
        \sigma_{\rm (+,o)(-,e)}^{\rm het} \coloneqq \sigma^{-1}_{\rm oo} \((C_{\rm o} C_{\rm e})^{-\frac{1}{2}}\bra{\beta} M^{\rm het}_{\rm oo} M^{\rm het}_{\rm (+,o)(-,e)} \ket{\beta} - \lambda^{\rm het}_{\rm oo}\lambda^{\rm het}_{\rm (+,o)(-,e)}\) = (\sigma_{\rm (+,o)(-,e)}^{\rm het})^*, \\
        \sigma_{\rm (-,o)(+,e)}^{\rm het} \coloneqq \sigma^{-1}_{\rm oo} \((C_{\rm o} C_{\rm e})^{-\frac{1}{2}}\bra{\beta} M^{\rm het}_{\rm oo} M^{\rm het}_{\rm (-,o)(+,e)} \ket{\beta} - \lambda^{\rm het}_{\rm oo}\lambda^{\rm het}_{\rm (-,o)(+,e)}\) = (\sigma_{\rm (-,o)(+,e)}^{\rm het})^*, \\
        \sigma^{\rm het}_{\rm (-,e)(-,e)} \coloneqq [\sigma_{\rm (-,e)(+,o)}^{\rm het}]^{-1} \((C_{\rm o} C_{\rm e})^{-\frac{1}{2}}\bra{\beta} M^{\rm het}_{\rm (+,o)(-,e)} M^{\rm het}_{\rm ee} \ket{\beta} - \lambda^{\rm het}_{\rm (+,o)(-,e)}\lambda^{\rm het}_{\rm ee}\) = (\sigma^{\rm het}_{\rm (-,e)(-,e)})^*,\\
        \sigma^{\rm het}_{\rm (+,e)(+,e)} \coloneqq [\sigma_{\rm (+,e)(-,o)}^{\rm het}]^{-1} \((C_{\rm o} C_{\rm e})^{-\frac{1}{2}}\bra{\beta} M^{\rm het}_{\rm (-,o)(+,e)} M^{\rm het}_{\rm ee} \ket{\beta} - \lambda^{\rm het}_{\rm (-,o)(+,e)}\lambda^{\rm het}_{\rm ee}\) = (\sigma^{\rm het}_{\rm (+,e)(+,e)} )^*,\\
        \Delta^{\rm het}_{\rm (+,o)(-,e)} \coloneqq \(C_{\rm e}^{-1}\left\|M^{\rm het}_{\rm (+,o)(-,e)}\ket{\beta}\right\|^2 - |\lambda^{\rm het}_{\rm (+,o)(-,e)}|^2 - |\sigma^{\rm het}_{\rm (+,o)(-,e)}|^2 \)^{\frac{1}{2}}, \\
        \Delta^{\rm het}_{\rm (-,o)(+,e)} \coloneqq \(C_{\rm e}^{-1}\left\|M^{\rm het}_{\rm (-,o)(+,e)}\ket{\beta}\right\|^2 - |\lambda^{\rm het}_{\rm (-,o)(+,e)}|^2 - |\sigma^{\rm het}_{\rm (-,o)(+,e)}|^2 \)^{\frac{1}{2}}, \\
        \Delta_{\rm (-,e)(-,e)}^{\rm het} \coloneqq \(C_{\rm e}^{-1}\left\|M^{\rm het}_{\rm ee} \ket{\beta}\right\|^2 - (\lambda^{\rm het}_{\rm ee})^2 - |\sigma^{\rm het}_{\rm (-,e)(-,e)}|^2 \)^{\frac{1}{2}},\\
        \Delta_{\rm (+,e)(+,e)}^{\rm het} \coloneqq \(C_{\rm e}^{-1}\left\|M^{\rm het}_{\rm ee} \ket{\beta}\right\|^2 - (\lambda^{\rm het}_{\rm ee})^2 - |\sigma^{\rm het}_{\rm (+,e)(+,e)}|^2 \)^{\frac{1}{2}}.
    \end{gather}
    Define the following two matrices $M'^{(0)}_{\rm 6d}$ and $M'^{(1)}_{\rm 6d}$.
    \begin{align}
        M'^{(0)}_{\rm 6d} &\coloneqq 
        \begin{pmatrix}
            1 & 0 & 0 & \Delta^{\rm het}_{\rm (+,o)(-,e)} & 0 & 0 \\
            0 & 1 & \sigma^{\rm het}_{\rm oo} & \sigma^{\rm het}_{\rm (+,o)(-,e)} & 0 & 0 \\
            0 & \sigma^{\rm het}_{\rm oo} & \kappa C_{\rm o} + \lambda^{\rm het}_{\rm oo} &  \kappa \sqrt{C_{\rm o}C_{\rm e}} + \lambda^{\rm het}_{\rm (+,o)(-,e)} & \sigma^{\rm het}_{\rm (-,e)(+,o)} & 0 \\
            \Delta^{\rm het}_{\rm (+,o)(-,e)} & \sigma^{\rm het}_{\rm (+,o)(-,e)} & \kappa \sqrt{C_{\rm o}C_{\rm e}} + \lambda^{\rm het}_{\rm (+,o)(-,e)} & \kappa C_{\rm e} + \lambda^{\rm het}_{\rm (-,e)(-,e)} - \gamma & \sigma^{\rm het}_{\rm (-,e)(-,e)} & \Delta^{\rm het}_{\rm (-,e)(-,e)} \\
            0 & 0 & \sigma^{\rm het}_{\rm (-,e)(+,o)} & \sigma^{\rm het}_{\rm (-,e)(-,e)} & 1 - \gamma & 0 \\
            0 & 0 & 0 & \Delta^{\rm het}_{\rm (-,e)(-,e)} & 0 & 1 - \gamma
        \end{pmatrix}, \\[1ex]
        M'^{(1)}_{\rm 6d} &\coloneqq 
        \begin{pmatrix}
            1 - \gamma & 0 & 0 & \Delta^{\rm het}_{\rm (-,o)(+,e)} & 0 & 0 \\
            0 & 1 - \gamma & \sigma^{\rm het}_{\rm oo} & \sigma^{\rm het}_{\rm (-,o)(+,e)} & 0 & 0 \\
            0 & \sigma^{\rm het}_{\rm oo} & \kappa C_{\rm o} + \lambda^{\rm het}_{\rm oo} - \gamma & \kappa \sqrt{C_{\rm o}C_{\rm e}} + \lambda^{\rm het}_{\rm (-,o)(+,e)} & \sigma^{\rm het}_{\rm (+,e)(-,o)} & 0 \\
            \Delta^{\rm het}_{\rm (-,o)(+,e)} & \sigma^{\rm het}_{\rm (-,o)(+,e)} & \kappa \sqrt{C_{\rm o}C_{\rm e}} + \lambda^{\rm het}_{\rm (-,o)(+,e)} & \kappa C_{\rm e} + \lambda^{\rm het}_{\rm ee} & \sigma^{\rm het}_{\rm (+,e)(+,e)} & \Delta^{\rm het}_{\rm (+,e)(+,e)} \\
            0 & 0 & \sigma^{\rm het}_{\rm (+,e)(-,o)} & \sigma^{\rm het}_{\rm (+,e)(+,e)} & 1 & 0 \\
            0 & 0 & 0 & \Delta^{\rm het}_{\rm (+,e)(+,e)} & 0 & 1 
        \end{pmatrix}.
    \end{align}
    Define a convex function
    \begin{equation}
        B^{\rm het}(\kappa,\gamma) \coloneqq \max\{\sigma_{\rm sup}(M'^{(0)}_{\rm 6d}), \sigma_{\rm sup}(M'^{(1)}_{\rm 6d})\}. \label{eq:B_kappa_gamma_het}
    \end{equation}
    Then, for $\kappa,\gamma\geq 0$, we have
    \begin{equation}
        M^{\rm het}[\kappa,\gamma] \leq B^{\rm het}(\kappa,\gamma) I_{AC}.
        \label{eq:bounded_by_B_het}
    \end{equation}
\end{cor}

\begin{proof}
    In the same way as Homodyne protocol, we have from Eqs.~\eqref{eq:tau_het_0_1} and \eqref{eq:W_A_het} that
    \begin{equation}
        \ket{u_\pm^{\rm het}(\omega_r)}_A = \cos\biggl(\frac{\theta_{\mu,\beta}^{\rm het}(\omega_r)}{2}\biggr)\ket{\pm}_A \pm \sin\biggl(\frac{\theta_{\mu,\beta}^{\rm het}(\omega_r)}{2}\biggr)\ket{\mp}_A.
    \end{equation}
    Combining this with Eqs.~\eqref{eq:het_integrated} and \eqref{eq:def_O}, we observe that
    \begin{align}
        &(\bra{+}_A\otimes\Pi_{\rm od}) M^{\rm het}_{\rm ph} (\ket{+}_A \otimes \Pi_{\rm od}) = (\bra{-}_A\otimes\Pi_{\rm od}) M^{\rm het}_{\rm ph} (\ket{-}_A \otimes \Pi_{\rm od}) = M^{\rm het}_{\rm oo},
        \label{eq:sandwiched_het_od_od} \\
        &(\bra{-}_A\otimes\Pi_{\rm ev}) M^{\rm het}_{\rm ph} (\ket{-}_A \otimes \Pi_{\rm ev}) = (\bra{+}_A\otimes\Pi_{\rm ev}) M^{\rm het}_{\rm ph} (\ket{+}_A \otimes \Pi_{\rm ev}) = M^{\rm het}_{\rm ee} 
        \label{eq:sandwiched_het_ev_ev}\\
        &(\bra{+}_A\otimes\Pi_{\rm od}) M^{\rm het}_{\rm ph} (\ket{-}_A \otimes \Pi_{\rm ev}) = \bigl[(\bra{-}_A\otimes\Pi_{\rm ev}) M^{\rm het}_{\rm ph} (\ket{+}_A \otimes \Pi_{\rm od})\bigr]^{\dagger} = M^{\rm het }_{\rm (+,o)(-,e)}
        \label{eq:sandwiched_het_plus_od_minus_ev}\\
        &(\bra{-}_A\otimes\Pi_{\rm od}) M^{\rm het}_{\rm ph} (\ket{+}_A \otimes \Pi_{\rm ev}) = \bigl[(\bra{+}_A\otimes\Pi_{\rm ev}) M^{\rm het}_{\rm ph} (\ket{-}_A \otimes \Pi_{\rm od})\bigr]^{\dagger} = M^{\rm het }_{\rm (-,o)(+,e)} 
        \label{eq:sandwiched_het_minus_od_plus_ev}
    \end{align}
    As can be seen from Eq.~\eqref{eq:het_integrated} as well as Eqs.~\eqref{eq:def_of_pi_fid}--\eqref{eq:def_of_M_kappa_gamma}, we have
    \begin{equation}
        M^{\rm het}[\kappa,\gamma] = \Pi_{AC}^{(+,\mathrm{od}), (-,\mathrm{ev})} M^{\rm het}[\kappa,\gamma]\, \Pi_{AC}^{(+,\mathrm{od}), (-,\mathrm{ev})} + \Pi_{AC}^{(-,\mathrm{od}),(+,\mathrm{ev})} M^{\rm het}[\kappa,\gamma]\, \Pi_{AC}^{(-,\mathrm{od}), (+,\mathrm{ev})},
    \end{equation}
    where $\Pi_{AC}^{(+,\mathrm{od}), (-,\mathrm{ev})}$ and $\Pi_{AC}^{(-,\mathrm{od}),(+,\mathrm{ev})}$ are defined in Eqs.~\eqref{eq:plus_od_minus_ev} and \eqref{eq:minus_od_plus_ev}.
    Then we apply Lemma \ref{lemma:operator_ineq} to the operators $\Pi_{AC}^{(+,\mathrm{od}), (-,\mathrm{ev})} M^{\rm het}[\kappa,\gamma]\, \Pi_{AC}^{(+,\mathrm{od}), (-,\mathrm{ev})}$ and $\Pi_{AC}^{(-,\mathrm{od}),(+,\mathrm{ev})} M^{\rm het}[\kappa,\gamma]\, \Pi_{AC}^{(-,\mathrm{od}), (+,\mathrm{ev})}$, respectively.  For $\Pi_{AC}^{(+,\mathrm{od}), (-,\mathrm{ev})} M^{\rm het}[\kappa,\gamma]\, \Pi_{AC}^{(+,\mathrm{od}), (-,\mathrm{ev})}$, using Eqs.~\eqref{eq:sandwiched_het_od_od}, \eqref{eq:sandwiched_het_ev_ev}, and \eqref{eq:sandwiched_het_plus_od_minus_ev}, we set
    \begin{align}
        \Pi_{\pm} &= \ket{\pm}\!\bra{\pm}_A \otimes \Pi_{\rm od(ev)}, \\
        M &= \Pi_{AC}^{(+,\mathrm{od}), (-,\mathrm{ev})} M^{\rm het}_{\rm ph} \Pi_{AC}^{(+,\mathrm{od}), (-,\mathrm{ev})} \\
        &= \ket{+}\!\bra{+}_A\otimes M^{\rm het }_{\rm oo} + \ket{-}\!\bra{-}_A\otimes  M^{\rm het}_{\rm ee}  +  \ket{+}\!\bra{-}_A\otimes M^{\rm het}_{\rm (+,o)(-,e)} + \ket{-}\!\bra{+}_A\otimes M^{\rm het}_{\rm (-,e)(+,o)}, \\
        \ket{\psi} &= \sqrt{\kappa} \ket{\phi_-}_{AC}, \\
        \alpha &= 1, \quad \gamma_+ = 0, \quad \gamma_{-} = \gamma,
    \end{align}
    where $\ket{\phi_-}_{AC}$ is defined in Eq.~\eqref{eq:phi_minus}.
    Since so-defined $M$ only has continuous spectrum, we can apply Lemma~\ref{lemma:operator_ineq} and obtain
    \begin{equation}
        \sigma_{\rm sup}\!\(\Pi_{AC}^{(+,\mathrm{od}), (-,\mathrm{ev})} M^{\rm het}[\kappa,\gamma]\, \Pi_{AC}^{(+,\mathrm{od}), (-,\mathrm{ev})}\) \leq \sigma_{\rm sup}(M'^{(0)}_{\rm 6d}). \label{eq:bounded_by_6d_0_het}
    \end{equation}
    We also apply Lemma \ref{lemma:operator_ineq} to $\Pi_{AC}^{(-,\mathrm{od}),(+,\mathrm{ev})} M^{\rm het}[\kappa,\gamma]\, \Pi_{AC}^{(-,\mathrm{od}), (+,\mathrm{ev})}$.  Using Eqs.~\eqref{eq:sandwiched_het_od_od}, \eqref{eq:sandwiched_het_ev_ev}, and \eqref{eq:sandwiched_het_minus_od_plus_ev}, we set
    \begin{align}
        \Pi_{\pm} &= \ket{\mp}\!\bra{\mp}_A \otimes \Pi_{\rm od(ev)}, \\
        M &= \Pi_{AC}^{(-,\mathrm{od}), (+,\mathrm{ev})} M^{\rm het}_{\rm ph} \Pi_{AC}^{(-,\mathrm{od}), (+,\mathrm{ev})} \\
        &= \ket{-}\!\bra{-}_A\otimes M^{\rm het}_{\rm oo}  + \ket{+}\!\bra{+}_A\otimes  M^{\rm het}_{\rm ee} + \ket{-}\!\bra{+}_A\otimes M^{\rm het}_{\rm (-,o)(+,e)} + \ket{+}\!\bra{-}_A\otimes M^{\rm het}_{\rm (+,e)(-,o)}, \\
        \ket{\psi} &= \sqrt{\kappa} \ket{\phi_+}_{AC}, \\
        \alpha &= 1, \quad \gamma_+ = \gamma, \quad \gamma_{-} = 0.
    \end{align}
    where $\ket{\phi_+}_{AC}$ is defined in Eq.~\eqref{eq:phi_plus}.
    Then, we observe 
    \begin{equation}
        \sigma_{\rm sup}\!\(\Pi_{AC}^{(-,\mathrm{od}),(+,\mathrm{ev})} M^{\rm het}[\kappa,\gamma]\, \Pi_{AC}^{(-,\mathrm{od}), (+,\mathrm{ev})}\) \leq \sigma_{\rm sup}(M'^{(1)}_{\rm 6d}).
        \label{eq:bounded_by_6d_1_het}
    \end{equation}
    Combining inequalities \eqref{eq:bounded_by_6d_0_het} and \eqref{eq:bounded_by_6d_1_het} completes the proof.
\end{proof}

Note that the elements of the matrices $M_{\rm 6d}^{(0)}$ and $M_{\rm 6d}^{(1)}$ in Eqs.~\eqref{eq:M_6d_0} and \eqref{eq:M_6d_1} depend linearly on the parameters $\kappa$ and $\gamma$.  Since the largest eigenvalue of a matrix is a convex function of its elements \cite{Horn1985,Overton1993}, the largest eigenvalues $\sigma_{\rm sup}(M_{\rm 6d}^{(0)})$ and $\sigma_{\rm sup}(M_{\rm 6d}^{(1)})$ are convex functions of $\kappa$ and $\gamma$, and so is $B^{\rm hom}(\kappa, \gamma)$ in Eq.~\eqref{eq:B_kappa_gamma_hom}.  The same fact can also be applied to $B^{\rm het}(\kappa, \gamma)$ in Eq.~\eqref{eq:B_kappa_gamma_het}.

\bibliographystyle{quantum}
\bibliography{refined_reverse}

\end{document}